\documentclass[runningheads]{llncs}

\usepackage[paperwidth=6.69in, paperheight=9.61in, text={4.81in,7.59in}, centering]{geometry}

\usepackage{cite}
\usepackage{amsmath,amssymb}
\usepackage{graphicx}
\usepackage{hyperref}
\usepackage{pgfplots}
\pgfplotsset{compat=1.17}

\newcommand{\ignore}[1]{}

\newcommand{\eqdef}{\stackrel{\mathit{def}}{=}}

\newcommand{\Reals}{\mathbb{R}}
\newcommand{\PosReals}{\Reals_{\ge 0}}

\newcommand{\Segments}{{\cal S}}
\newcommand{\segnorm}[1]{|\!| #1 |\!|}
\newcommand{\subseg}[3]{{#1}{[#2,#3]}}

\newcommand{\pre}[1]{{^\bullet\!{#1}}}
\newcommand{\post}[1]{{{#1}^\bullet}}

\newcommand{\segment}[1]{{#1}.\mathit{seg}}
\newcommand{\Pos}[1]{\mathit{Pos}_{#1}}
\newcommand{\edgeride}[4]{{#2} \xrightarrow{{#3}}_{#1} {#4}}
\newcommand{\ride}[4]{{#2} \stackrel{#3}\rightsquigarrow_{#1} {#4}}

\newcommand{\crossing}{\bowtie}

\newcommand{\att}[2]{\mathit{#2}_{#1}}
\newcommand{\po}[1]{\att{#1}{p}}

\newcommand{\speed}[1]{\att{#1}{v}}

\newcommand{\limit}[1]{\att{#1}{\pi}}

\newcommand{\cd}[1]{\att{#1}{cd}}
\newcommand{\fo}[1]{\att{#1}{fo}}

\newcommand{\cfg}{\mathit{cfg}}
\newcommand{\scfg}{\mathit{scfg}}

\newcommand{\Vehicles}{{\cal C}}
\newcommand{\Objects}{{\cal H}}

\newcommand{\Brake}{\mathsf{B}}
\newcommand{\DeltaV}{\mathsf{C}}

\newcommand{\follow}[1]{\mathsf{follow}({#1})}
\newcommand{\timeTo}[1]{\mathsf{tt}({#1})}

\newcommand{\clf}{\mathit{cl}}
\newcommand{\initkw}{\mathbf{init}}
\newcommand{\dokw}{\mathbf{if}}
\newcommand{\odkw}{\mathbf{fi}}

\newcommand{\mytextcircled}[1]
           {\raisebox{.5pt}{\textcircled{\raisebox{-.9pt} {#1}}}}

\newcommand{\timeToRed}[1]{\att{#1}{ttred}}
\newcommand{\timeToGreen}[1]{\att{#1}{ttgreen}}

\begin{document}

\usetikzlibrary{math}

\title{Compositionally Safe Construction of Autonomous Driving Systems}

\author{Marius Bozga\orcidID{0000-0003-4412-5684} \and
  Joseph Sifakis\orcidID{0000-0003-2447-7981}}

\authorrunning{M. Bozga, J. Sifakis}

\institute{Univ. Grenoble Alpes, CNRS, Grenoble
  INP\footnote{Institute of Engineering Univ. Grenoble Alpes},
  VERIMAG, 38000 Grenoble, France
  \email{\{Marius.Bozga,Joseph.Sifakis\}@univ-grenoble-alpes.fr}\\
  \url{http://www-verimag.imag.fr/}}

\maketitle

\begin{abstract}
  Developing safe autonomous driving systems is a major scientific and
technical challenge. Existing AI-based end-to-end solutions do not
offer the necessary safety guarantees, while traditional systems
engineering approaches are defeated by the complexity of the problem.
We study a method for building compositionally safe autonomous driving
systems, based on the assumption that the capability to drive boils
down to the coordinated execution of a given set of driving
operations. The assumption is substantiated by a compositionality
result considering that autopilots are dynamic systems receiving a
small number of types of driving configurations as input, each
configuration defining a free space in its neighborhood. It is shown
that safe driving for each type of configuration in the corresponding
free space, implies safe driving for any possible scenario under some
easy-to-check conditions concerning the transition between
configurations.  The designed autopilot comprises distinct control
policies one per type of driving configurations, articulated in two
consecutive phases. The first phase consists of carefully managing a
potentially risky situation by virtually reducing speed, while the
second phase consists of exiting the situation by accelerating.
The autopilots designed use for their predictions simple functions
characterizing the acceleration and deceleration capabilities of the
vehicles. They cover the main driving operations, including entering a
main road, overtaking, crossing intersections protected by traffic
lights or signals, and driving on freeways.  The results presented
reinforce the case for solutions that incorporate mathematically
elegant and robust decision methods that are safe by construction.

  \begin{keywords}
    autonomous driving systems, assume-guarantee techniques,
    compositionally safe construction
  \end{keywords}
\end{abstract}

\section{Introduction}\label{sec:introduction}
The development of trustworthy autonomous driving
systems (ADS) is today a major scientific and technical challenge that
could have a considerable economic and societal impact. It will also
be an important step towards building systems with human-level
intelligence. ADS need to combine extensive cognitive capabilities for
situation awareness and management of multiple objectives, while
meeting stringent requirements for safety and adaptability.

There are two main approaches to building ADS. The first uses
human-designed multi-module architectures, e.g., Autoware
\cite{Autoware-1.1}, Apollo \cite{Apollo-3.0}. In general, perception
and prediction modules make extensive use of neural networks, while
planning and control modules rely on traditional systems engineering
techniques. The other approach is to build end-to-end autopilots that
treat all driving as a single AI model, based on the vehicle's various
sensors and producing a trajectory from which steering angle and
acceleration/deceleration signals can be calculated
e.g., \cite{ChittaPJYRG23, ShaoW00022, HuCGMGYKCS22}. The advantages
and disadvantages of both approaches are well known. Multi-module
solutions are advantageous in terms of interpretability and
verifiability, but their intrinsic complexity can lead to a lack of
module composability, in particular compound errors and loss of
information. End-to-end AI-based solutions have the advantage of
simplicity and lower development costs, but they can be less
deterministic and suffer from the well-known problems of AI systems,
such as non-explainability and anomalies \cite{abs-2306-16927,
  KurakinGB18, StoreyL0P22}.  At present, none of these approaches is
satisfactory, as evidenced by the current debate on the deployment of
autonomous cars in urban environments, and real-world experience which
confirms that ADS still have a long way to go before offering the
necessary safety guarantees \cite{Cummings23}.

Our work adopts the multi-module approach, focusing on prediction,
planning and control. It assumes that perception problems are solved
by a separate module that generates a model of the external
environment in the form of an annotated map containing relevant
information about its state. This includes the geometric
characteristics of roads, signaling equipment and obstacles around the
vehicle, with their positions and kinematic attributes.

The formalization of maps and their underlying concepts has been the
object of numerous studies including proposals of standards such as
\cite{ASAMOpenDRIVE-1.6.0,ASAMOpenScenario-1.0.0}. Modeling maps as
extended graphs is a common idea adopted by many works, one of the
main problem being to build maps by composition of heterogeneous
elements and data e.g.,
\cite{PoggenhansPJONK18,ButzHHORSZ20,BeetzB18}. In \cite{BozgaS22}, we
study a multilevel semantic framework for the description of maps and
their properties, based on metric graphs. The framework allows maps to
be specified as a composition of building blocks such as different
types of roads and junctions. It also shows how the traffic rules for
each of these building blocks induce corresponding safety constraints
to be respected by autopilots.

The design of autopilots for autonomous vehicles has been the subject
of studies originating in robotics and control theory. Some studies
focus on architectures and their hierarchical structure, from the
fastest tasks at the lowest level to the slowest tasks at higher
levels, e.g. \cite{Staron21}.  In the field of architectures, the NIST
4D/RCS reference model \cite{Albus22} provides a basis for the design,
engineering and software integration of intelligent systems for
unmanned ground vehicles. The reference model integrates planners for
a set of tasks, each dedicated to a specific objective for predefined
operations. It implicitly assumes that the driving ability can be
summed up as the combination of skills required to perform a set of
elementary driving operations.  Our work is also based on such an
assumption, which allows us to decompose the autopilot behavior into a
set of specific control policies that can be designed and verified
separately. Validating this assumption poses two problems: the first
is to prove the correctness of the composition, i.e., that if the
specific control policies are safe so is the resulting autopilot
behavior. The other is completeness, i.e., that the set of control
policies considered is sufficient to drive safely in all
circumstances.

All the traditional control and game-theoretic approaches are relevant
to this study. These include PID (proportional-integral-derivative),
MPC (model-predictive control) and LPV (linear parameter-varying)
control approaches, whose application to ADS is however hampered by
problems of inherent complexity and limitations concerning their
ability to realistically take into account the safety and optimality
constraints of autonomous driving. A detailed analysis of these
approaches is presented in \cite{ChenLHXTLHTLWCZW23}. Of particular
interest are correct-by-construction techniques, where the autopilot
or some of its features are constructed from a set of properties
derived from system requirements. There is considerable work on
controller synthesis from a set of system properties, usually
expressed in linear temporal logic, see for example
\cite{Kress-GazitP08,SchwartingAR18,arxiv.2203.14110,WongpiromsarnKF11,WongpiromsarnTM12}. These
techniques have been extensively studied in the field of control. They
consist of restricting the controllable behavior of the system
interacting with its environment so that the desired properties are
satisfied. However, their practical value is limited by their high
computational cost, which depends in particular on the type of
properties and the complexity of the system's behavior.

An alternative to synthesis is to achieve correctness by construction
as a result of composing component properties expressed as
``assume-guarante'' contracts specifying a causal relationship between
components and their environment: if the environment satisfies the
``assume'' part of the contract, the state of the component will satisfy
the ``guarantee'' part, e.g.,
\cite{BenvenisteCNPRR18,ChatterjeeH07,Meyer92}.  A prerequisite for
contract-based design is the decomposition of global system
requirements into contracts, and the existence of suitable techniques
for their implementation \cite{abs-1909-02070}. There are a number of
theoretical frameworks that apply mainly to continuous or synchronous
systems, particularly for analysis and verification purposes
\cite{MavridouKGKPW21,SaoudGF21,abs-2012-12657}. They suffer from
computational limitations because, in the general case, they involve
the symbolic solution of fixed-point equations
\cite{MavridouKGKPW21}. Furthermore, they only apply to systems with a
static architecture, which excludes reconfigurable dynamic systems,
such as autonomous systems.

This paper builds on previous results \cite{BozgaS23} concerning a
correct coordination method for autonomous driving systems that allows
the construction of a Runtime coordinating a set of autonomous
vehicles based on their positions and kinematic attributes. The method
considers \emph{systems of} ADS (SADS) as dynamic systems involving
vehicles moving in a static environment modelled as maps
\cite{BozgaS22}. A key idea is that the vehicles are bounded to move
within their corresponding allocated free spaces computed by the
Runtime so as the whole system satisfies given properties including
collision avoidance and traffic rules.  The main result is that, given
a set of traffic rules, the Runtime can calculate authorizations for
the coordinated vehicles in such a way that the resulting system
behavior respects the traffic rules by construction. It is established
by showing that the composition of assume-guarantee contracts is an
inductive invariant that entails SADS safety.

This work adopts the same perspective of correct-by-construction, but
considers that there is no overall coordination of vehicles. Instead,
they drive autonomously, as in real life, with limited knowledge of
their physical environment and surrounding obstacles.  The approach
adopted is based on the assumption that the complexity of ADS
construction can be mastered by factoring it into three dimensions.

\textbf{Locality of context:} ADS operate in complex environments
taking on a wide variety of configurations, each of which can affect
the behavior of the system. Therefore, ADS safety is strongly
dependent on the context in which vehicles operate. The traffic
infrastructure can be seen as the composition of a finite number of
patterns comprising different types of roads and junctions with their
signaling equipment. We can therefore imagine that a vehicle's safety
policy is the composition of elementary policies, each of which is
used to drive safely according to the corresponding basic road
patterns.

\textbf{Locality of knowledge:} A vehicle’s driving policy is based
only on local knowledge of the SADS state due to limited
visibility. It must therefore drive safely, taking into account the
obstacles closest to it, delimited by a visibility zone. In this way,
the collective behavior of vehicles in a SADS can be understood and
analyzed as the composition of smaller sets of vehicles grouped
according to proximity and visibility criteria.

\textbf{Rights-based responsibility:} ADS are agents in a distributed
system where each agent is responsible for managing a space within its
planned route, defined by traffic rules. The rules ensure that if each
vehicle drives safely in the free space dynamically determined by its
rights, the whole system is safe. This principle of rights-based
responsibility \cite{abs-1708-06374,abs-2206-03418} greatly simplifies
the validation problem, since all that needs to be demonstrated is
that each vehicle drives safely in its own free space.

We show that the above decomposition principles reduce the general
problem to the construction of a vehicle autopilot capable of driving
responsibly in a limited number of contexts and configurations
involving a relatively small number of other vehicles and objects.

We call \emph{driving configuration} the input of an autopilot
generated by its perception function after analysis and interpretation
of the information provided by the sensors. Thus, a driving
configuration defines the state of the environment of a vehicle as
well as the applicable traffic rules. The state of the environment
includes the positions and kinematic attributes of other vehicles, as
well as information on signaling equipment and its status. A driving
configuration also delimits an area where the vehicle has to act
responsibly assuming that the mobile agents in this area behave in
accordance with traffic regulations.  A vehicle's driving
configuration depends on the visibility determined by various factors
in its environment, including topology and physical obstacles, as well
as weather and light conditions.  A simple analysis shows that, as a
vehicle moves, its autopilot reacts to inputs that are changes in the
state of its environment, characterized by three different types of
driving configurations:

\begin{enumerate}
\item \emph{road configurations} where there are no crossroads in the
  vehicle's area of visibility, and the autopilot is tasked with
  taking into account the obstacles in its route ahead;
\item \emph{merging configurations} when the vehicle's route joins a
  road or a lane where oncoming vehicles have a higher priority and it
  must therefore give way to these vehicles;
\item \emph{crossing configurations} where the vehicle's route crosses
  a junction accessible to other vehicles, and therefore, the vehicle
  must comply with the traffic rules applicable in this context.
\end{enumerate}

A key idea of this work is that to guarantee safety of a SADS, it is
sufficient to ensure safe driving of the vehicles involved for each
type of driving configurations.  The main result is obtained by
compositionality, in two main steps.

First, we provide safe driving policies for different contexts in each
type of driving configurations. Given a reference vehicle, called ego
vehicle, safe driving for a configuration requires a specific driving
operation to overcome potential conflicts while respecting the
applicable traffic rules. We argue that critical situations can be
characterized by configurations involving, in addition to the ego
vehicle, an oncoming vehicle whose route may intersect that of the ego
vehicle and a front vehicle located after the intersection on the
ego's route. The vehicle’s driving policy should take into account its
dynamic characteristics, in particular its braking and acceleration
capacity, as well as the relationships with surrounding obstacles, in
particular their speeds and their distances from locations where
collisions can occur.

Secondly, based on the assumption that over the course of a journey,
the changing perceived environment is a succession of driving
configurations, we show that if a vehicle can drive safely for each
type of configuration, the resulting behavior is safe under certain
simple temporal conditions. These guarantee the consistency of
transitions between the configurations making up the vehicle's route.

The paper is structured as follows.  Section~\ref{sec:approach}
presents the overall approach including definition of basic concepts,
the types of driving configurations and corresponding control policy
principles.  Section~\ref{sec:policies} provides control policies for
three driving configuration types and demonstrates their safety.
Section~\ref{sec:correctness} shows the correctness of the method by
proving compositionality of the policies for the different types of
driving configurations.  Section~\ref{sec:discussion} concludes by
summarizing the main results obtained and discussing future
developments with a view to their effective application.

\section{The Approach}\label{sec:approach}
\subsection{Environment modeling with maps} \label{sec:approach:maps}

Following the idea presented in \cite{BozgaS22}, we build contiguous
road segments from a set $\Segments$ equipped with a partial
concatenation operator $\cdot : \Segments \times \Segments{}
\rightarrow \Segments \cup \{\bot\}$, a length norm $\segnorm{.} :
\Segments \rightarrow \PosReals$ and a partial sub-segment extraction
operator $\subseg{.}{.}{.}:\Segments \times \PosReals \times \PosReals
\rightarrow \Segments \cup \{\bot\}$.  Given a segment $s$,
$\segnorm{s}$ represents its length and $\subseg{s}{a}{b}$ for $0 \le
a<b \le \segnorm{s}$, represents the sub-segment starting at length
$a$ from its origin and ending at length $b$.  Segments can be used to
represent roads at different levels of abstraction.  The highest level
may ignore the form for the segment and give only its length. The
lowest level can be a two-dimensional area. An intermediate level can
be a curve showing the form of the road and making abstraction of its
width.

We use metric graphs $G \eqdef (U,E,\Segments)$ to represent maps,
where $U$ is a finite set of \emph{vertices}, $\Segments$ is a set of
segments and $E \subseteq U \times \Segments^\star \times U$ is a
finite set of \emph{edges} labeled by \emph{non-zero length} segments
(denoted $\Segments^\star$).  For an edge $e=(u,s,u') \in E$ we denote
$\pre{e} \eqdef u$, $\post{e} \eqdef u'$, $\segment{e} \eqdef s$.
We call a metric graph \emph{connected} if a path exists between any
pair of vertices.

We consider the set $\Pos{G} \eqdef U \cup \{(e,a) \mid e \in E,~ 0
\le a \le \segnorm{\segment{e}}\}$ of the \emph{positions} defined by
a metric graph.  Note that positions $(e,0)$ and $(e,
\segnorm{\segment{e}})$ are considered equal respectively to positions
$\pre{e}$ and $\post{e}$.  We denote by $\edgeride{G}{p}{s}{p'}$ the
existence of an $s$-labelled \emph{edge route} between succeeding
positions $p=(e,a)$ and $p'=(e,a')$ in the same edge $e$ whenever $0
\le a < a' \le \segnorm{\segment{e}}$ and $s =
\subseg{\segment{e}}{a}{a'}$.  Moreover, we denote by
$\ride{G}{p}{s}{p'}$ the existence of an $s$-labelled \emph{route}
between arbitrary positions $p$, $p'$, that is, $\ride{G}{}{}{} \eqdef
(\edgeride{G}{}{}{})^+$ the transitive closure of edge routes.
Finally, we denote by $p'-p$ the \emph{distance} from position $p$ to
position $p'$ defined as 0 whenever $p = p'$ or as the minimum length
among all segments labeling routes from $p$ to $p'$ or as $+\infty$ if
no such route exists.  Whenever $G$ is fixed in the context, we will
omit the subscript $G$ for positions $\Pos{G}$ and routes
$\edgeride{G}{}{}{}$ or $\ride{G}{}{}{}$.

A connected metric graph $G=(U,E,\Segments)$ can be interpreted as a map,
structured into roads and junctions, subject to additional assumptions:
\begin{itemize}
\item We restrict to metric graphs which are 2D-consistent
  \cite{BozgaS22}, meaning intuitively they can be drawn in the
  2D-plane such that the geometric properties of the segments are
  compatible with the topological properties of the graph. In
  particular, if two distinct paths starting from the same vertex $u$,
  meet at another vertex $u'$, the coordinates of $u'$ calculated from
  each path are identical.  For the sake of simplicity, we further
  restrict to graphs where distinct vertices are located at distinct
  points in the plane, and moreover, where no edge is self-crossing
  (meaning actually that distinct positions $(e,a)$ of the same edge
  $e$ correspond to distinct points).
\item We consider that if the segments of two edges $e_1$ and $e_2$
  intersect at distances $a_1$ and $a_2$ from their starting points,
  then for the respective positions of the point of intersection we
  have $(e_1, a_1) = (e_2, a_2)$.

  We now define the \emph{junctions} of a map as the classes of an
  equivalence relation $\crossing$ on edges, obtained as the
  transitive closure of the relation generated by pairs of edges
  $(e_1,e_2)$ such that: either $e_1$ and $e_2$ intersect or the
  endpoint of $e_1$ and $e_2$ is the same vertex
  ($\post{e_1}=\post{e_2}$).  That is, a junction is a connected
  sub-graph such that for each edge $e_1$ there exists another edge
  $e_2$ with the same endpoint or intersecting $e_1$.  Note that if we
  remove the junctions of a map we obtain sets of roads, where a road
  is a sequence of non-intersecting edges whose initial vertices have
  in-degree 1.  We assume that junctions are equipped with additional
  signals to regulate traffic on the edges, e.g., traffic lights, stop
  signs, etc.
\end{itemize}

In the remainder of the paper, we consider a fixed metric graph
$G=(U,E,\Segments)$ altogether with the junction relationship
$\crossing$.  Also, we extend the junction relationships from edges to
their associated positions, that is, consider $(e_1,a_1) \crossing
(e_2,a_2) \eqdef e_1 \crossing e_2$ whenever $e_1 \crossing e_2$.
Finally, we denote by $r_1 \uplus r_2$ the property that routes $r_1$,
$r_2$ in $G$ are \emph{non-intersecting}, that is, their sets of positions
are disjoint and moreover not belonging to the same junction(s),
except for endpoints.

\subsection{SADS as dynamic systems}\label{sec:approach:ads}

A system of ADS (SADS) is a dynamic system involving a set of \emph{vehicles}
$\Vehicles$, a set of regulatory \emph{signals} $\Objects$, and a map
$G$ that represents the environment where the signals are located and
the vehicles can move. We use the term \emph{obstacle} to refer to a
vehicle or a signal.

The state $q$ of a SADS is the union of the states of its
vehicles and signals, $q \eqdef q_\Vehicles \cup q_\Objects$, where
$q_\Vehicles \eqdef \{q_c\}_{c \in \Vehicles}$ and $q_\Objects \eqdef
\{q_h\}_{h \in \Objects}$ where:
\begin{itemize}
\item The state of a vehicle $c$, is a tuple $q_c \eqdef \langle
  \po{c}, s_c, \speed{c}, V_{c}, \ldots \rangle$, where $\po{c}$ is
  the position of $c$ on the map, $s_c$ is a segment labeling a route
  on the map starting at $\po{c}$, $\speed{c}$ is the speed of $c$,
  $V_c$ is its speed limit enforced at position $\po{c}$.
\item The state of a signal $h$, is a tuple $q_h \eqdef \langle
  \po{h}, \mathit{type}_h, \ldots \rangle$, where $\po{h}$ is its
  position on the map and $\mathit{type}_h$ denotes its type.  We
  consider signals of the following types:
  \begin{itemize}
  \item stop or yield signs guarding junctions, with a critical
    distance attribute $cd_h$ defining the length of the route segment
    they are protecting,
  \item traffic lights guarding junctions, with a critical distance
    attribute $cd_h$, a color attribute $\mathit{color}_h$ taking
    values \textit{red}, \textit{yellow} or \textit{green}, and
    \emph{time-to-color} attributes $\timeToRed{h}$,
    $\timeToGreen{h}$; note that $\timeToRed{h}=0$ if
    $\mathit{color}_h=red$ and otherwise $\timeToRed{h}>0$ is
    keeping the strictly positive duration until $\mathit{color}_h$
    became $red$ (same for $\timeToGreen{h}$),
  \item speed limits, with an attribute $V_h$ defining the enforced
    speed limit.
  \end{itemize}
\end{itemize}
A SADS evolves from an initial state $q[0] \eqdef q_\Vehicles[0] \cup
q_\Objects[0]$ and through states $q[t] \eqdef q_\Vehicles [t] \cup
q_\Objects[t]$, with $q[t] \stackrel{\Delta t}{\rightarrow} q[t +
  \Delta t]$ where $\Delta t$ is an adequately chosen time step. The
latter can be the period of the autopilots.

In the context of a SADS state, we introduce some additional notations
on the underlying metric graphs $G$.  For a vehicle $c$ and
non-negative value $d$ we denote by $\po{c} +_c d$ the unique position
$p$ located ahead on the route of $c$ at distance $d$, that is,
formally satisfying $\ride{}{\po{c}}{s_c[0, d]}{p}$.  Moreover, for
any such two positions $p_1 = \po{c} +_c d_1$, $p_2 = \po{c} +_c d_2$
we write $p_1 \sim_c p_2'$ iff $d_1 \sim d_2$ for any $\sim \in
\{<,\le,\ge,>\}$.

\subsubsection{Visibility zone}

For a reference vehicle called \emph{ego} vehicle, we define the
concept of \emph{driving configuration} characterized by its state and
the states of the
obstacles in its visibility zone. The visibility zone is defined by an
area of the map around the ego vehicle using two types of parameters
(Fig.~\ref{fig:vista}):
\begin{itemize}
\item The \emph{frontal visibility} of the ego vehicle on its route up
  to a front distance $fd(q_e)$ is delimiting the interval in which
  the ego’s autopilot can perceive the obstacles on its route. The
  distance $fd(q_e)$ depends on factors such as road curvature,
  obstacles in view and weather conditions at position $\po{e}$.
\item The \emph{lateral visibility} of the ego vehicle when its route
  meets a road or lane $r$, is the distance from the junction point at
  which the ego vehicle can perceive vehicles arriving from $r$. We
  denote by $ld_r(q_e)$ this distance.
\end{itemize}

These parameters determine the visibility zone by points on the route of
the ego vehicle and also on the possibly intersecting roads from which
vehicles can arrive.

  \begin{figure}[htbp]
    \centering
    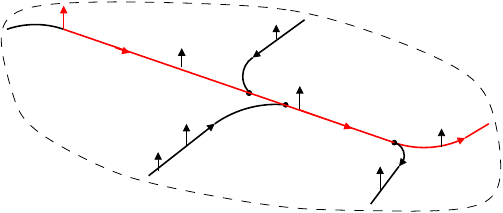
    \caption{\label{fig:vista}{Driving configuration of an ego vehicle with position
  $\po{e}$ and route $s_e$ and two arriving vehicles}.}
  \end{figure}

\subsubsection{Driving configuration}

Given a SADS, a driving configuration for an ego vehicle with state $q_e$ is a triple
$\cfg \eqdef \langle q_e, q_F, q_A \rangle$ where
\begin{itemize}
\item $q_e$ is the state of the ego vehicle.
\item $q_F$ is the ordered set of the states of front obstacles $F
  \eqdef \{f_1,\ldots, f_n\} \subseteq \Vehicles \cup \Objects$
  located on the route of the ego vehicle in its visibility zone such
  that $\po{e} \le_e \po{f_1} <_e \cdots <_e \po{f_n} = \po{e} +_e
  fd(q_e)$.  We consider that the last visible obstacle $f_n$ is a
  fictitious vehicle at the front visibility limit.
\item $q_A$ is the set of the states of arriving vehicles $A \subseteq
  \Vehicles$, at most one per road joining the route of the ego
  vehicle in its frontal visibility zone.  If there is no real vehicle
  arriving on the road $r$ within the corresponding visibility
  distance $ld_r(q_e)$, then we consider a fictitious vehicle $a_r$
  with state $q_{a_r} = \langle \po{a_r}, s_{a_r}, \speed{a_r},
  V_{a_r} \rangle$ defined such that (i) $\po{a_r}$ is the position on
  the segment representing the road $r$ at distance $ld_r(q_e)$ from
  the junction point, (ii) $s_{a_r}$ is the segment between $\po{a_r}$
  and the junction point and (iii) $\speed{a_r} = V_{a_r} =$ the speed
  limit enforced at $p_{a_r}$ on this route.  Thus, for vehicles $a_r$
  arriving from a road $r$, the position $p_r$ of the junction point
  is such that:
  $$\exists d_e \le fd(q_e).~\exists d_r \le ld_r(q_e).~
  \po{e} +_e d_e = \po{a_r} +_{a_r} d_r = p_r.
  $$
\end{itemize}

Note that while the driving configuration $\cfg$ includes the states
of all visible obstacles in the ego vehicle's route, it only includes
the state of a single arriving vehicle joining or crossing its route
within the frontal visibility limit.  The precautionary principle
requires us to consider fictitious vehicles at the limits of the
visibility zone. These fictitious vehicles can be a frontal obstacle
at distance $fd(q_e)$, or a vehicle arriving from a road or lane
joining its route at distance $ld_r(q_e)$. In this way, visibility
constraints are implicitly taken into account in a configuration.

\subsubsection{SADS as dynamic systems}

A SADS with $m$ vehicles, is a dynamic system using a map $G$ and
having states $q[t]$ that can change after time $\Delta t$ to
$q[t+\Delta t]$.  At the top of Fig.~\ref{fig:ads} we show a
decomposition of the SADS as a dynamic system with autopilots one for
each vehicle $c$, which knowing the global system state $q[t]$ at
time $t$ and the map $G$, compute their new state $q_c[t+ \Delta
  t]$ at time $t+ \Delta t$.


  \begin{figure}[htbp]
    \centering
    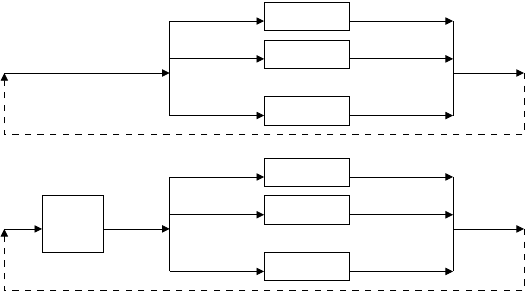
    \caption{\label{fig:ads}{SADS as a dynamic system composed of the
    environment and the autopilots}.}
  \end{figure}

At the bottom of Fig.~\ref{fig:ads}, this architecture is further
refined with the autopilots and a component representing the
environment of the vehicles. The environment component receives the
states of the vehicles at the end of the cycle and computes the global
state of the SADS using the map and the knowledge of the state of its
signals. Furthermore, it produces for the $c$ vehicle autopilot the
corresponding driving configurations $\cfg_c(q[t])$ from the global state $q[t]$ taking
into account visibility parameters.


%
%
%

\subsection{The three basic driving configurations types}\label{sec:approach:vista-types}

As explained in the Introduction, we consider that the autopilot of
the ego vehicle receives basic types of configurations as input, each
requiring specific operations implemented by the corresponding control
policies (Fig.~\ref{fig:vista-types}):

\mytextcircled{1} \textbf{Road configurations} are of the form $\cfg \eqdef
\langle q_e, q_F, \emptyset \rangle$ where there is no junction in the
frontal visibility area of the ego vehicle. Note that road configurations can
be simplified in the following manner.  Consider the ordered set of
the front obstacles $F$ and remind that the last visible obstacle is a
fictitious vehicle at the front visibility limit.  Clearly, these
obstacles are either vehicles or signals (e.g., speed limits).  The
ego vehicle is responsible for driving safely in the space on its
route until to the closest visible vehicle, hence all the obstacles
after this vehicle can be omitted as irrelevant.

  \begin{figure}[htbp]
    \centering
    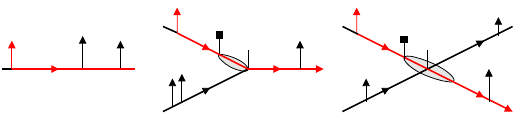
    \caption{\label{fig:vista-types}{Road, merging and crossing configurations - critical sections in gray}.}
  \end{figure}

\mytextcircled{2} \textbf{Merging configurations} describe situations where
the route of the ego vehicle merges into a main road. They are of the
form $\cfg \eqdef \langle q_e, q_F, q_a \rangle$ such that there exists
a merging position $p$ satisfying
$$\exists d_e,d_a,d.~ \po{e} +_e d_e = \po{a} +_a d_a = p, 
s_e[d_e,d_e+d] = s_a[d_a,d_a+d]$$
and moreover $\po{e} \le_e \po{f_1} \le_e p <_e \po{f_2}$ where $f_1$
is a yield or stop sign.  That is, all the elements of $F$ other than
$f_1$ are located after the merging point $p$.  A merging configuration can be
simplified by replacing $q_F$ by $q_{F'}$ where $ \langle q_e, q_{F'},
\emptyset \rangle$ is the simplified road configuration for $\langle q_e,
q_F,\emptyset \rangle$.  To avoid collision, we assume that the signal
$f_1$ is at critical distance $\cd{f_1}$ on the route of the ego
vehicle before the merging point.

Note that the ego vehicle deals with merging configurations when its route
merges into a higher-priority road or when it changes lanes. In
particular, overtaking involves two successive merging operations: one
consists of moving from the initial lane to an adjacent lane, the
other of returning to the initial lane after a phase of driving in a
straight line to ensure that it is far enough away from the overtaken
vehicle.
  
\mytextcircled{3} \textbf{Crossing configurations} describe situations where
the route of the ego vehicle crosses a main road. They are of the form
$\cfg \eqdef \langle q_e, q_F, q_A \rangle$ such that there exists a
crossing position $p$ satisfying
$$\exists d_e, d_a.~\po{e}+_e d_e = \po{a} +_a d_a = p$$
for all arriving vehicles $a \in A$. We assume that $f_1$
is a signal such as a traffic light or a stop sign located before
$p$ and protecting a critical distance $\cd{f_1}$, that is, the length
of $s_e$ in the intersection.  A simplified crossing configuration $\scfg \eqdef
\langle q_e, q_{F'}, q_A \rangle$ is such that $\langle q_e, q_{F'},
\emptyset \rangle$ is a simplified road configuration.  We can define sub-types
of crossing configurations. One is when $f_1$ is a traffic light, another when
$f_1$ is a yield or a stop sign.

Based on this decomposition into types of configurations, the autopilot can be
architectured as the serial composition of two components
(Fig.~\ref{fig:arch-autopilot}):
\begin{itemize}
\item a \emph{Configuration Manager} component that receives configurations and
  produces simplified configurations with their type and relevant parameters
  such as the protecting signals for crossing configurations and the
  positions where routes may intersect;
\item a \emph{Control Policy Manager} that receives the simplified
  configuration and their type, applies the corresponding control policy and
  produces the next state of the ego vehicle, sent to the environment
  component responsible for composing the next configuration for each
  vehicle. The Control Policy Manager also generates commands to the
  drive-by-wire platform to modify the speed by $\Delta v$ and the
  steering angle by $\Delta \phi$, in distance $\Delta d$. Note that
  the proposed controllers only calculate the speed variation
  $\Delta v$ and distance $\Delta d$. We assume that the corresponding
  steering angle $\Delta \phi$ can be estimated from the curvature of
  the road segment at the current position and the distance
  $\Delta d$.
\end{itemize}

  \begin{figure}[htbp]
    \centering
    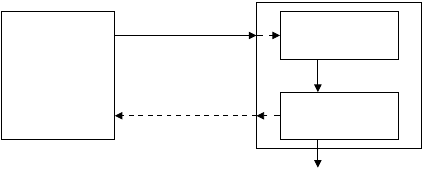
    \caption{\label{fig:arch-autopilot}{Architecture of the autopilot}.}
  \end{figure}

\subsection{Characteristic parameters of driving configurations}\label{sec:approach:vista-parameters}

A driving configuration for the ego vehicle describes a situation involving a
potential obstacle on its route and calling for operation implemented
by a specific control policy. In its simplest and general form, a
driving configuration involves (Fig.~\ref{fig:vista-parameters}):
\begin{itemize}
\item the ego vehicle at position $\po{e}$, with its route $s_e$ and
  speed $\speed{e}$
\item an arriving vehicle $a$ at speed $\speed{a}$, located at
  position $\po{a}$ on its route $s_a$, which encounters the route
  $s_e$ at position $p$
\item a signal $h$ at position $\po{h}$ on the route $s_e$, located
  before and protecting the obstacle $p$ with critical distance $\cd{h}$
  i.e., $\po{h} <_e p \le_e \po{h} +_e \cd{h}$
\item a front vehicle $f$ at position $\po{f}$ on the route $s_e$
  located after the obstacle at a distance $d_f=\po{f}-p$.
\end{itemize}

  \begin{figure}[htbp]
    \centering
    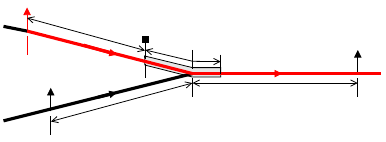
    \caption{\label{fig:vista-parameters}{A driving configuration and its
  characteristic parameters}.}
  \end{figure}

It is easy to check that the three different types of driving configurations match
this model.  For road configurations the obstacle is just a speed limit
signal.  Merging configurations cover two cases. The first is when the
obstacle can be a yield sign and the road of the ego vehicle merges
into a road of higher priority where the arriving vehicle is
traveling. The second case is where the ego vehicle moves from one
lane to another in which the arriving vehicle is traveling. In both
cases, we assume that there is a front vehicle after the merging
point.  Crossing configurations cover also two cases. The first is when the
obstacle is the traffic lights protecting an intersection, in which
case arriving vehicles are irrelevant. The second case is an
intersection where the crossing ego vehicle faces a yield or stop
sign.

\section{Control Policies for Configuration Types}\label{sec:policies}

\subsection{Assumptions about the vehicle's dynamics}
\label{sec:policies:dynamics}

As explained, an autopilot for a given driving configuration, issues commands that aim
at controlling the speed of the vehicle based on its knowledge of the
dynamics of the vehicle as an electromechanical system.  This
knowledge is important for predicting vehicle behavior and its ability
to execute given commands aimed at modifying its kinematic state. To
avoid detailed modeling of the vehicle as a dynamic system, see for
example \cite{Wenfei21}, we assume that we know for each vehicle the
next two functions that are sufficient to decide feasibility of control:

\begin{enumerate}
\item The \emph{braking function} $\Brake(v)$ that gives the distance
  needed to brake from speed $v$ to speed 0.  We assume the braking
  function is strict and monotonic, that is, $B(0)=0$ and $B(v_1) >
  B(v_2)$ for any $v_1 > v_2$.



\item The \emph{speed control function} $\DeltaV(v, \mathit{VL})$
  which gives the pair $(\Delta v, \Delta d)$ of the greatest speed
  variation $\Delta v$ and the associated distance traveled $\Delta d$
  after $\Delta t$, if any achievable from speed $v$ such that $0 \le
  v+\Delta v$ and $v + \Delta v, \Delta d$ compliant to speed limit
  constraints $\mathit{VL}$ as explained below.
\end{enumerate}
It is clear that the speed limit constraints impose changes in speed
which must be compatible with vehicle's dynamics characterized by the
two functions $\Brake$ and $\DeltaV$.  We assume that the speed limit
constraints $\mathit{VL}$ encountered by a vehicle on its route are
specified as sequences of pairs $\langle(d_i,V_i)\rangle_{i=0,n}$ of
distance $d_i$ and speed limit $V_i$ (see Fig.~\ref{fig:dv}), meaning
that the speed limit $V_i$ is enforced in the interval $[d_i,d_{i+1})$
for all $i\in[0,n]$ (and where implicitly $0 = d_0 < d_1 < \ldots < d_n <
d_{n+1} \eqdef +\infty$).  To comply with speed limit constraints, the
above-mentioned functions must meet the following requirements:
\begin{itemize}
\item First, the $\DeltaV$ function is only defined for input
  arguments $v$ and $\mathit{VL}\eqdef\langle(d_i,V_i)\rangle_{i=0,n}$ that
  satisfy the \emph{controllability} condition $\Brake(v) \le
  d_i + \Brake(V_i)$, that is, the vehicle is able to slow down from speed $v$ to
  $V_i$ in distance less than $d_i$, for all $i \in [0,n]$.
\item Second, in this case, the output pair $(\Delta v, \Delta d)
  \eqdef \DeltaV(v, \mathit{VL})$ shall satisfy, for all $i\in [0,n]$:
  \begin{itemize}
  \item $v + \Delta v \le V_i$ whenever
    $d_i \le \Delta d < d_{i+1}$,
  \item $\Delta d + \Brake(v + \Delta v) \le d_i + \Brake(V_i)$
    whenever $\Delta d \le d_i$,
  \item $\Delta d \le d_i$ whenever $V_i = 0$.
  \end{itemize}
  That is, the new speed $v + \Delta v$ satisfies the
  \emph{controllability} condition with respect to speed limits still
  ahead after traveling the distance $\Delta d$ and moreover, $\Delta
  d$ never exceeds the distance to a speed 0 limit.
\end{itemize}
These joint requirements of $\Brake$ and $\DeltaV$ functions allow us
later to establish safety invariance properties relating the speed of
the vehicle, its braking capacity and its distance to front
obstacles (speed limit signals, moving vehicles, etc).  Their
effective use to define speed control policies for specific
driving situations will be explained in the next sections.

As an example, we show how these functions can be defined in a simple
case. Assume for the sake of simplicity that the autopilot can only
select, during every $\Delta t$ period, among different constant
acceleration/deceleration values in the interval $[-b_{max},
  a_{max}]\subseteq \Reals$, that is, between a maximal deceleration
rate $-b_{max}$ and a maximal acceleration rate $a_{max}$.  In this
case the braking function $\Brake_{ex}$ and the speed control function
$\DeltaV_{ex}$ could be defined as follows:
\begin{multline*}
  \Brake_{ex}(v) \eqdef \mbox{if } v < b_{max} \Delta t \mbox{ then }  v \Delta t / 2 \\
  \mbox{ else } v \Delta t - b_{max} \Delta t^2 / 2 + \Brake_{ex}(v - b_{max} \Delta t) \\
\DeltaV_{ex}(v, \langle(d_i,V_i)\rangle_{i=0,n}) \eqdef (a^* \Delta t, v \Delta t + a^* \Delta t^2/ 2) \mbox{ where}  \\
a^*  \eqdef \min\nolimits_{i=0,n}  \max \{ a \in [-b_{max},a_{max}]~\mid \\
v + a\Delta t \ge 0,~ \Delta d = v \Delta t + a \Delta t^2 / 2, \\
d_i  >0  \mbox{ and } \Delta d + \Brake_{ex}(v + a \Delta t) \le d_i + \Brake_{ex}(V_i) \mbox{ or } \\
d_i  =0  \mbox{ and } v + a \Delta t \le V_i \}
\end{multline*}

That is, the function $\Brake_{ex}$ defines the distance
needed to decrease from speed $v$ to speed 0 when only constant
decelerations are used for braking during fixed periods $\Delta t$.
The function $\DeltaV_{ex}$ defines the speed variation as the
maximal variation satisfying \emph{all} the speed limit constraints.

It is an easy check that, if the inputs given to the function
$\DeltaV_{ex}$ satisfy the controllability condition $\Brake_{ex}(v)
\le d_i + \Brake_{ex}(V_i)$ then, the sets of
acceleration/deceleration values over which the maximum is computed
are not empty, for every speed limit constraint $(d_i,V_i)$.  That is,
the set of values $a$ compliant with the $i$th constraint contains, at
least:
\begin{itemize}
\item $a_{max}$ whenever $d_i > 0$ and $v \Delta t + a_{max} \Delta t^2/2 + \Brake_{ex}(v + a_{max} \Delta t)
   \le d_i + \Brake_{ex}(V_i)$,
\item $0$ whenever $d_i > 0$ and $v \Delta t + \Brake_{ex}(v) \le d_i + \Brake(V_i)$ or $d_i = 0$ or $v \le V_i$,
\item $-b_{max}$ whenever $b_{max} \Delta t \le v$, 
\item $(V_i-v)/\Delta t$ whenever $0 <  v - V_i < b_{max} \Delta t$.
\end{itemize}

For example, Fig.~\ref{fig:dv} represents the speed limit
constraints $\mathit{VL}=$ $\langle(0 m$, $100 km/h)$, $(40 m, 50 km/h)$,
$(140 m, 0 km/h)\rangle$ with three thick horizontal lines, respectively
$100$ $km/h$ during $40$ meters (interval $[0,40)$), followed by $50
km/h$ for $100$ meters (interval $[40,140)$), and then stop at $0
km/h$ at $140$ meters (interval $[140,-)$).  For each speed limit,
below the solid line of the same color is the corresponding
controllability region, assuming $b_{max} = -3.4 m/s^2$ and $\Delta t
= 1s$. Moreover, the area below the dotted line of the same color is
the part of the corresponding controllability region from which the
vehicle is controllable after applying the maximal acceleration
$a_{max} = 2.5 m/s^2$ for $\Delta t$.

The figure shows also the possible speed variations and the associated
distance for respectively initial speeds of $30$, $60$ and $90 km/h$.
For initial speed of $30$ and $60 km/h$ the vehicle is
controllable with respect to the three speed limits, whereas for
initial speed $90 km/h$ it is not possible to respect the speed limit
of $50 km/h$. The figure shows in green for $30$ and $60 km/h$,
possible choices of speed variations and the corresponding traveled
distance during $\Delta t$. Note that, whereas every acceleration
preserves controllability for initial speed $30 km/h$, the maximal
acceleration $a_{max}=2.5 m/s^2$ for $60 km/h$ results in a speed
exceeding the controllability limit.

\begin{figure}[t]
  \resizebox{\columnwidth}{!}{
\begin{tikzpicture}

  \tikzmath{\a=2.5; \b=3.4; \xt=1;
    \vl1=27.79; \vl2=13.89; \vl3=0.0;
    \d1=0.0; \d2=40.0; \d3=140.0; \d4=160.0;
    \v0=8.3;}
  \begin{axis}[grid,
      xlabel={$d ~~ (m)$}, ylabel={$v ~~ (km/h)$},
      width=\textwidth,  height=0.35\textheight,
      enlarge x limits=false,
      legend style={}, legend pos=north east
    ]
    
    \addplot[brown, domain=\d1:\d2, smooth, style={ultra thick}] {3.6 * \vl1};
    \addplot[blue, domain=\d2:\d3, smooth, style={ultra thick}] {3.6 * \vl2};
    \addplot[magenta, domain=\d3:\d4, smooth, style={ultra thick}] {3.6 * \vl3};

    \addplot[blue, domain=\vl2:22.0, smooth] ( {\d2 - (x * x - \vl2 * \vl2) / (2 * \b)}, {3.6 * x} );
    \addplot[magenta, domain=\vl3:29.0, smooth] ( {\d3 - (x * x - \vl3 * \vl3) / (2 * \b)}, {3.6 * x} );
    

    \addplot[brown, domain=\d1:\d2, dashed] {3.6 * (\vl1 - \a * \xt) };
    \addplot[blue, domain=\vl2-\a*\xt:16.8, dashed] 
    ( {\d2 - ((x + \a * \xt) * (x + \a * \xt) - \vl2 * \vl2) / (2 * \b) - x * \xt - (\a * \xt * \xt) / 2}, {3.6 * x} );
    \addplot[blue, domain=\d2-13:\d3, dashed] {3.6 * (\vl2 - \a * \xt) };    
    \addplot[magenta, domain=\vl3:25.5, dashed]
    ( {\d3 - ((x + \a * \xt) * (x + \a * \xt) - \vl3 * \vl3) / (2 * \b) - x * \xt - (\a * \xt * \xt) / 2}, {3.6 * x} );

    \addplot[green, domain=0:\xt, style={ultra thick}] ( { \v0 * x + \a * x * x / 2}, {3.6 * ( \v0 + \a * x) });
    \addplot[green, domain=0:\xt, style={ultra thick}] ( { \v0 * x + \a * x * x / 4}, {3.6 * ( \v0 + \a * x / 2) });
    \addplot[green, domain=0:\xt, style={ultra thick}] ( { \v0 * x}, {3.6 * ( \v0 ) });
    \addplot[green, domain=0:\xt, style={ultra thick}] ( { \v0 * x - \b * x * x / 4}, {3.6 * ( \v0 - \b * x / 2) });
    \addplot[green, domain=0:\xt, style={ultra thick}] ( { \v0 * x - \b * x * x / 2}, {3.6 * ( \v0 - \b * x) });
 
    \addplot[red, domain=0:\xt, style={ultra thick}] ( { 2 * \v0 * x + \a * x * x / 2}, {3.6 * ( 2 * \v0 + \a * x) });
    \addplot[green, domain=0:\xt, style={ultra thick}] ( { 2 * \v0 * x + \a * x * x / 4}, {3.6 * ( 2 * \v0 + \a * x / 2) });
    \addplot[green, domain=0:\xt, style={ultra thick}] ( { 2 * \v0 * x}, {3.6 * ( 2 * \v0 ) });
    \addplot[green, domain=0:\xt, style={ultra thick}] ( { 2 * \v0 * x - \b * x * x / 4}, {3.6 * ( 2 * \v0 - \b * x / 2) });
    \addplot[green, domain=0:\xt, style={ultra thick}] ( { 2 * \v0 * x - \b * x * x / 2}, {3.6 * ( 2 * \v0 - \b * x) });

    \addplot[red, domain=0:\xt, style={ultra thick}] ( { 3 * \v0 * x + \a * x * x / 2}, {3.6 * ( 3 * \v0 + \a * x) });
    \addplot[red, domain=0:\xt, style={ultra thick}] ( { 3 * \v0 * x + \a * x * x / 4}, {3.6 * ( 3 * \v0 + \a * x / 2) });
    \addplot[red, domain=0:\xt, style={ultra thick}] ( { 3 * \v0 * x}, {3.6 * ( 3 * \v0 ) });
    \addplot[red, domain=0:\xt, style={ultra thick}] ( { 3 * \v0 * x - \b * x * x / 4}, {3.6 * ( 3 * \v0 - \b * x / 2) });
    \addplot[red, domain=0:\xt, style={ultra thick}] ( { 3 * \v0 * x - \b * x * x / 2}, {3.6 * ( 3 * \v0 - \b * x) });
    
    \legend{ {100km/h at [0,40)}, {50km/h at [40,140)},  {0km/h at [140,-)}  }
  \end{axis}

\end{tikzpicture}
}
\caption{\label{fig:dv}Speed control illustration according to speed limits $\mathit{VL}$}
\end{figure}

\subsection{The $\mathsf{follow}$ policy and travel time prediction}
\label{sec:policies:follow}

Using the braking $\Brake$ and speed control $\DeltaV$ functions, we
now define the primitive $\mathsf{follow}$ policy and the
related travel time prediction function $\mathsf{tt}$.  These
primitives are used next to define the more sophisticated control
policies for types of driving configurations.

The $\follow{\cfg,p}$ policy takes as input a position $p$ located in
the current configuration $\cfg$ of the ego vehicle ahead on its route. This
policy controls the speed of the ego vehicle so that it follows safely
the position $p$ i.e., by keeping the distance $p-\po{e}$ as small as
possible while remaining safe.

First, the policy observes the speed limits enforced on the ego
vehicle route in its configuration $\cfg$ till the position $p$.  To do so, it
relies on the function $\mathsf{speed\mbox{-}limits}(\cfg,p)$ defined
next to compute the speed limit constraints $\mathit{VL}$ in the
appropriate format defined in subsection~\ref{sec:policies:dynamics}.
\begin{multline*}
  \mathsf{speed\mbox{-}limits} (\langle q_e,q_F,q_A\rangle, p) \eqdef \\ 
  \mbox{\textbf{let} } \langle h_1, \ldots, h_k \rangle \leftarrow
  \langle h \in F ~|~ type_h=\textit{speed limit}, \po{h} <_e p \rangle, \\
  \mbox{\textbf{return} } \langle (0,V_e),(\po{h_1}-\po{e},V_{h_1}), \ldots, 
    (\po{h_k}-\po{e},V_{h_k}),(p-\po{e},0) \rangle 
\end{multline*}
That is, the current speed limit $V_e$ is applied at the current
position, the next speed limit $V_{h_1}$ is applied from its position
$\po{h_1}$ (if any), etc, and the speed limit 0 is applied at position $p$.

Second, the $\follow{\cfg,p}$ policy uses the speed control function
$\DeltaV$ given the current speed $\speed{e}$ and the speed limit
constraints $\mathit{VL}$ on the route to compute the commands
$\Delta v$ and $\Delta d$ for the driving platform (as explained in
subsection~\ref{sec:approach:vista-types} and illustrated in
Fig.~\ref{fig:arch-autopilot}):
\begin{multline*} \follow{\cfg, p} \eqdef \\
  \mbox{\textbf{let} } \mathit{VL} \leftarrow
  \mathsf{speed\mbox{-}limits}(\cfg, p), \mbox{\textbf{drive} } \DeltaV(v_e, \mathit{VL})
 \end{multline*}
Recalling that the speed control function $\DeltaV$ is a partial
function, note that $\follow{\cfg,p}$ executes correctly only if
$\DeltaV$ is called with parameters satisfying the controllability
property.  To this end, the next proposition establishes the
pre/postconditions for the correct execution of $\follow{\cfg,p}$ in
terms of $\cfg$ and $p$ according to the assumptions about the $\Brake$
and $\DeltaV$ functions.

\begin{proposition}\label{prop:speed-controllability}
  Let $\langle h_1, ..., h_k \rangle$ be the ordered set of speed
  limit signs from the front obstacles $F$ in $\cfg$ before $p$.
  The policy $\follow{\cfg,p}$ executes correctly only if
  \begin{enumerate}
  \item\label{it1:prop:speed-controllability}
    $\speed{e} \le V_e$, that is, current speed is lower than the
    current speed limit,
  \item\label{it2:prop:speed-controllability}
    $\po{e} +_e \Brake(\speed{e}) \le_e \po{h_i} +_e \Brake(V_{h_i})$, that is, speed
    limit preservation with respect to any speed limit $h_i$ located
    such that $\po{e} <_e \po{h_i} <_e p$,
  \item\label{it3:prop:speed-controllability}
    $\po{e} +_e \Brake(\speed{e}) \le_e p$, that is, safety
    distance preservation to the position $p$.
  \end{enumerate}
  In this case, the speed variation $\Delta v$ and the
  distance traveled $\Delta d$ satisfy, for any speed limit signal $h_i$ as above (where taking implicitly
  $\po{h_{k+1}} \eqdef p$):
  \begin{enumerate}\setcounter{enumi}{3}
  \item\label{it5:prop:speed-controllability}
    $\speed{e} + \Delta v \le V_e$ if $\po{e} +_e \Delta d <_e \po{h_1}$,
  \item\label{it6:prop:speed-controllability}
    $\speed{e} + \Delta v \le V_{h_i}$ if $\po{h_i} \le_e \po{e} +_e \Delta d <_e \po{h_{i+1}}$,
  \item\label{it7:prop:speed-controllability}
    $\po{e} +_e \Delta d +_e \Brake(\speed{e} + \Delta v) \le_e
    \po{h_i} +_e \Brake(V_{h_i})$ if $\po{e} +_e \Delta d \le_e \po{h_i}$,
  \item\label{it8:prop:speed-controllability}
    $\po{e} +_e \Delta d +_e \Brake(\speed{e} + \Delta v) \le_e p$.
  \end{enumerate}
\end{proposition}
\begin{proof}
    The $\follow{\cfg,p}$ policy calls the speed control function
    $\DeltaV$ with the current speed $\speed{e}$ and the speed limit
    constraints $\mathit{VL} \eqdef \langle (0, V_e), \langle
    (\po{h_i} - \po{e}, V_{h_i}) \rangle_{i\in[1,k]}, (p - \po{e},
    0)\rangle$ to obtain the speed variation $\Delta v$.  Then, the
    conditions
    \ref{it1:prop:speed-controllability}--\ref{it3:prop:speed-controllability}
    above are simply restating the controllability conditions, that
    is, the preconditions guaranteeing a valid return value $(\Delta
    v, \Delta d) \eqdef \DeltaV(\speed{e}, \mathit{VL})$, relative to
    the current speed and the position of the ego vehicles with
    respect to front obstacles in the configuration, respectively:
    \begin{itemize}
    \item $\Brake(\speed{e}) \le 0 + \Brake(V_e)$ (condition \ref{it1:prop:speed-controllability}),
    \item $\Brake(\speed{e}) \le \po{h_i} - \po{e} +
      \Brake(V_{h_i})$ for all $i \in [1,k]$ (condition \ref{it2:prop:speed-controllability}),
    \item $\Brake(\speed{e}) \le p - \po{e} + \Brake(0)$ (condition \ref{it3:prop:speed-controllability})
    \end{itemize}
    In a similar manner, the conditions \ref{it5:prop:speed-controllability}--\ref{it8:prop:speed-controllability} are restating the
    postcondition of the successful call, respectively:
    \begin{itemize}
    \item $\speed{e} + \Delta v \le V_e$ if $0 \le \Delta d < \po{h_1} - \po{e}$ (condition \ref{it5:prop:speed-controllability}),
    \item $\speed{e} + \Delta v \le V_{h_i}$ if $\po{h_i} - \po{e} \le
      \Delta d < \po{h_{i+1}} - \po{e}$ for some $i\in[1,k]$ (condition \ref{it6:prop:speed-controllability}),
    \item $\Delta d + \Brake(\speed{e} + \Delta v) \le \po{h_i} -
      \po{e} + \Brake(V_{h_i})$ for all $i\in [1,k]$ such that
      $\Delta d \le \po{h_i} - \po{e}$ (condition \ref{it7:prop:speed-controllability}),
    \item $\Delta d + \Brake(\speed{e} + \Delta v) \le p - \po{e} +
      \Brake(0)$ (condition \ref{it8:prop:speed-controllability})
      because the last limit constraint enforces the speed 0 at
      distance $p - \po{e}$.
    \end{itemize}
\end{proof}


Finally, consider a typical driving configuration $\cfg$ involving a front vehicle $f$
at position $\po{f}$ and a signal $h$ located at position $\po{h}$,
protecting a critical area situated between the ego vehicle and $f$
(as in Fig.~\ref{fig:vista-parameters}).  Usually, the decision for
the ego vehicle to cross or not the position $\po{h}$ between $\po{e}$
and $\po{f}$ depends on the time needed to go from $\po{e}$ to
$\po{h}+_e \cd{h}$ when safely following $f$.  This travel time can be predicted
if we assume the ego vehicle will drive continuously according to the
$\follow{\cfg,p}$ policy at least until it reaches the position
$\po{h}+_e \cd{h}$.



That is, as $\follow{\cfg,p}$ is a deterministic policy, we can
effectively predict the time to travel from $\po{e}$ to any position
$p'$ such that $\po{e} <_e p' \le_e p$ by using the function
$\timeTo{\cfg,p',p}$ defined below:
\[ \begin{array}{l}
 \timeTo{\cfg, p',p} \eqdef v \leftarrow v_e,~ \mathit{VL} \leftarrow \mathsf{speed\mbox{-}limits}(\cfg,p) \\
  ~~~~ x \leftarrow p' - \po{e}, ~ t \leftarrow 0 \\
  ~~~~ \mbox{\textbf{while }} x > 0 \mbox{\textbf{ do}} \\
  ~~~~ ~~~~ (\Delta v, \Delta d) \leftarrow \DeltaV(v, \mathit{VL}), \\
  ~~~~ ~~~~ v \leftarrow v + \Delta v, ~ x \leftarrow x - \Delta d,~ t \leftarrow t + \Delta t, \\
  ~~~~ ~~~~ \mathit{VL} \leftarrow \langle (\max(d_i-\Delta d, 0),V_i) ~|~ \\
  ~~~~ ~~~~ \hspace{2cm} (d_i,V_i)\in \mathit{VL},~ \Delta d \le d_{i+1} \rangle_{} \\
  ~~~~ \mbox{\textbf{return} }t
\end{array} \]
This function ``simulates'' driving according to $\follow{\cfg,p}$ and
sums up the $\Delta t$ steps until the ego vehicle travels the
distance $x = p' - \po{e}$.  Note that, in this computation we assume
that the leading position $p$ is not changing.  In typical situations,
$p$ may correspond to the position $\po{f}$ of a front vehicle $f$
moving forward.  In any case, however, consequent predictions for the
traveling time to $p'$ can either remain unchanged or decrease (but
never increase) over time.

\subsection{Control policies for driving configuration types}
\label{sec:policies:policies}
Each driving configuration requires a specific operation when the ego vehicle
approaches the obstacle. The aim of the operation is to clear the
obstacle safely, respecting traffic regulations and, of course,
avoiding accidents with the arriving vehicle and the front vehicle.
The operation corresponding to a configuration is logically characterized by
scenarios comprising two successive phases:

\begin{enumerate}
\item A \emph{caution phase} during which the ego vehicle approaches
  the obstacle, reducing its speed if necessary, and waiting for
  conditions to be favorable to clear the obstacle, e.g., approaching
  a merge, approaching a crossing or remaining in the same lane before
  overtaking.
\item A \emph{progress phase} during which the ego vehicle clears the
  obstacle after checking that there is no risk of collision with the
  arriving vehicle or the vehicle in front, e.g., to overtake a
  vehicle, enter a main road, cross an intersection. The progress
  phase therefore, consists of moving as quickly as possible to avoid
  collision with the arriving vehicle, while retaining the possibility
  of avoiding a collision with the front vehicle.
\end{enumerate}

An important issue in autopilot construction is the interplay between
cautious behavior and progress. An extremely cautious autopilot may be
safe but if it neglects the opportunity to progress it can degrade
performance resulting in bad road occupancy and possibly
deadlocks. Here are a few examples of over-cautious control policies:
driving at low speed on a freeway; stopping before a yield sign even
if the priority road is clear; not overtaking a slow vehicle in front,
such as a truck, when the ego vehicle's performance allows it and the
outside lane is clear; stopping before a green light.

In our method, we make a clear distinction between cautious driving
and progress by providing a clearance condition that enables the
transition between the two phases.  We define control policies as
processes run by the autopilot to control the speed of the vehicle by
generating every $\Delta t$ a corresponding speed change $\Delta v$.
They are defined on the basis of the iterative application of the
primitive $\mathsf{follow}$ policy.  Control policies are written as
simple programs $\mathsf{pol}_T$, one for every type of configuration $T$,
with guarded commands of the form
\begin{multline*}\mathsf{pol}_T(\cfg) \eqdef [\initkw \mapsto ...  ] \\
\dokw~ g_1 \mapsto
\mathsf{pol}_1(\cfg,...) \mid \cdots \mid g_k \mapsto
\mathsf{pol}_k(\cfg,...) ~\odkw
\end{multline*}
Their execution should be understood in relation to the architecture
of the autopilot (see Fig.~\ref{fig:arch-autopilot}) proposed in
subsection~\ref{sec:approach:vista-types}.  Every $\Delta t$, as long
as the input configurations $\cfg$ sent to the autopilot have the same type $T$,
the corresponding policy $\mathsf{pol}_T$ executes the command whose
guard is true.  A guard is a state predicate depending on the states
of the vehicles and signals of the input configuration.  Whenever the type of
the input configuration $\cfg$ changes to some $T' \not= T$ at some step $\Delta
t$, the next policy $\mathsf{pol}_{T'}$ is first (re)initialized and
then takes over with the execution of the commands as usual.

\subsection{Policy for road configurations}
\label{sec:policies:road}

The $\mathsf{road}$ policy for a simplified road configuration $\cfg=\langle
q_e,q_F,\emptyset \rangle$ is the direct application of
$\follow{\cfg,\po{f}}$, where $f$ is the last (leading) front obstacle
in $F$ i.e., either a real or fictitious vehicle before the frontal
visibility limit of the ego vehicle in $\cfg$.

\begin{align*}
  \mathsf{road}(\cfg) \eqdef  \dokw~~ \mathit{true} \mapsto \follow{\cfg, \po{f}} ~~\odkw
\end{align*}

\subsection{Policies for merging configurations}
\label{sec:policies:merging}

\subsubsection{Merging with yield sign}

The \textsf{merge-yield} policy applies for a simplified merging configuration
$\cfg = \langle q_e,q_F,q_a \rangle$ as illustrated in
Fig.~\ref{fig:merging-yield}.

In the caution phase, the ego vehicle approaches according to
$\follow{\cfg,\po{h}}$ where $h$ is the yield sign protecting the
merging point.  That is, the ego vehicle ensures $\Brake(\speed{e})
\le d_e$ and can stop if needed at critical distance $\cd{h}$ before the
merging point.  In the progress
phase, the ego vehicle moves towards the merging point according to
$\follow{\cfg,\po{f}}$ where $f$ is the leading front vehicle in $F$.

  \begin{figure}[htbp]
    \centering
    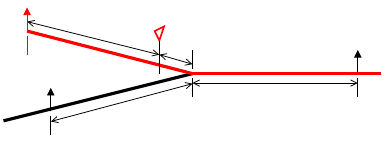
    \caption{\label{fig:merging-yield}{Merging configuration with yield sign}.}
  \end{figure}

The clearance condition allowing to switch from the caution to the
progress phase is defined as follows.  For the arriving vehicle $a$,
we assume that it can drive at the maximal speed limit $V_a$ on its
road and it is at distance $d_a$ from the merging point. So, if the
ego vehicle decides to progress and so switch to the
$\follow{\cfg,\po{f}}$ policy, the time to reach the merging point would
be at most $\timeTo{\cfg,\po{h}+_e\cd{h},\po{f}}$.  Within this time the
arriving vehicle will have traveled distance $V_a \cdot
\timeTo{\cfg,\po{h}+_e\cd{h},\po{f}}$. So the remaining distance from the
merging point will be $d_a - V_a \cdot \timeTo{\cfg,
  \po{h}+_e\cd{h},\po{f}}$. This distance should be large enough for a
safe brake, therefore, $ \Brake(V_a) \le d_a - V_a \cdot \timeTo{\cfg,
  \po{h}+_e\cd{h},\po{f}}$.

The \textsf{merge-yield} policy is therefore specified as follows.
Note that the switching between the two phases is controlled by the
boolean clearance flag $\clf$ initially set to $\mathit{false}$ and
updated continuously during the caution phase.
\begin{align*}
  \mathsf{merge\mbox{-}}&\mathsf{yield}(\cfg) \eqdef \\
  \initkw & \mapsto  ~\clf \leftarrow \mathit{false}, ~h \leftarrow \mathsf{yield\mbox{-}sign}(F) \\
   \dokw~~  \neg \clf & \mapsto  ~\follow{\cfg,\po{h}}, ~ \\
   & \clf \leftarrow ( V_a \cdot \timeTo{\cfg, \po{h}+_e\cd{h},\po{f}} + \Brake(V_a) \le d_a) \\
  \mid ~~~~ \clf & \mapsto  ~\follow{\cfg,\po{f}}  ~~\odkw
\end{align*}

\subsubsection{Lane change}

The lane change configuration $\cfg = \langle q_e, q_F, q_A\rangle$ is
illustrated in Fig.~\ref{fig:merging-lane-change}.  In the caution
phase, the ego vehicle follows the front vehicle $f_1$ on its lane.
The ego vehicle can initiate the progress phase of the lane change
whenever the following conditions hold, respectively:
\begin{enumerate}
\item the arriving vehicle $a_2$ can stop safely behind ego's
  position, $\Brake(V_{a_2}) \le d_{a_2}$,
\item the ego vehicle can stop safely behind the front vehicle $f_2$
  on the other lane, that is, $\Brake(\speed{e}) \le d_{f_2}$,
\item neither the front vehicle $f_1$ nor the back vehicle $a_1$ are
  currently performing a lane change.
\end{enumerate}

  \begin{figure}[htbp]
    \centering
    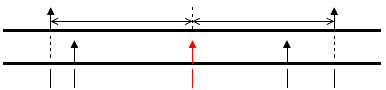
    \caption{\label{fig:merging-lane-change}{Lane change configuration}.}
  \end{figure}

These conditions provide the necessary guarantees for the ego vehicle
to (begin to) drive safely in the transition between the two lanes,
as required for the lane change.  To check condition (3) we assume
that vehicles signal lane changes, for example by using flashing
lights.
 
In the progress phase, the ego vehicle follows the front vehicle $f_2$
while also performing effectively the lane change. This policy is
essentially the same as $\follow{\cfg, \po{f_2}}$ but also includes
moving safely sideways into the corresponding lane as soon as possible
and activating the flashing lights.  When the lateral movement is
completed, the policy falls back to $\follow{\cfg,\po{f_2}}$. The lane
change policy is formalized as follows.
\begin{align*}
  \mathsf{lane\mbox{-}}&\mathsf{change}(\cfg) \eqdef 
  \initkw \mapsto ~\clf \leftarrow \mathit{false} \\
  \dokw ~~ \neg\clf & \mapsto ~\follow{\cfg,\po{f_1}}, ~ \\
  & \clf \leftarrow (\Brake(V_{a_2}) \le d_{a_2} \mbox{ and } \Brake(\speed{e}) \le d_{f_2}) \\
  \mid ~~~~ \clf & \mapsto ~\follow{\cfg, \po{f_2}}  ~~\odkw
\end{align*}

\subsection{Policies for crossing configuration}
\label{sec:policies:crossing}

\subsubsection{Crossing with yield sign}

The crossing configuration with yield sign is illustrated in
Fig.~\ref{fig:crossing-yield}.  It includes the ego vehicle
traveling at speed $\speed{e}$ at a distance $d_e$ from the point of
intersection with a main road protected by a yield sign $h$. In this
operation, the ego vehicle should approach cautiously moderating its
speed until it decides to progress if there is no risk of collision
with some arriving vehicle at distance $d_a$ and with allowed maximal
speed $V_a$. Additionally, it should avoid collision with front
vehicles at distance $d_f$ after the intersection. We suppose that
around the intersection point there is a critical area delimited by a
critical distance $\cd{h}$ such that the presence of two vehicles in this
area is considered a potential accident.

  \begin{figure}[htbp]
    \centering
    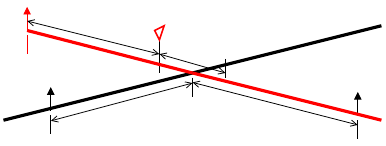
    \caption{\label{fig:crossing-yield}{Crossing configuration with yield sign}.}
  \end{figure}

For the caution phase, the ego vehicle is driving according to
$\follow{\cfg,\po{h}}$, that is, guaranteeing $\Brake(\speed{e}) \le d_e$.  In the
progress phase, the ego vehicle is driving according to
$\follow{\cfg,\po{f}}$.  Hence, the time needed by the ego vehicle to cross
and get out of the critical section is $\timeTo{\cfg, \po{h}+_e\cd{h},\po{f}}$. In this
time, the arriving vehicle will have traveled maximal distance $V_a
\cdot \timeTo{\cfg, \po{h}+_e\cd{h},\po{f}}$. So, the remaining space should be enough
to brake if needed to avoid collision in the critical area. This gives
the clearance condition $V_a \cdot \timeTo{\cfg, \po{h}+_e\cd{h},\po{f}} + \Brake(V_a)
\le d_a$ for switching between the two phases.  The crossing with
yield policy is therefore formalized as follows:
\begin{align*}
  \mathsf{cross\mbox{-}}&\mathsf{yield}(\cfg) \eqdef \\
  \initkw & \mapsto ~\clf \leftarrow \mathit{false}, h \leftarrow \mathsf{yield\mbox{-}sign}(F) \\
  \dokw~~  \neg ~\clf & \mapsto ~\follow{\cfg,\po{h}}, ~ \\
  & \clf \leftarrow (V_a \cdot \timeTo{\cfg, \po{h} +_e \cd{h},\po{f}} + \Brake(V_a) \le d_a) \\
  \mid ~~~~ \clf & \mapsto ~\follow{\cfg,\po{f}} ~~ \odkw
\end{align*}

\subsubsection{Traffic lights}

The crossing configuration $\cfg = \langle q_e, q_F, q_A \rangle$ illustrated in
Fig.~\ref{fig:crossing-traffic-light} involves the ego vehicle
approaching a traffic-light protected intersection.  It
also involves a signal $h$ of type traffic light with a state
variable $\mathit{color}_h$ taking values \textit{red},
\textit{yellow} and \textit{green}.

We assume that we know the duration $T_y$ of the yellow
light. Furthermore, we assume that the traffic lights of the
intersection have an ``\emph{all red}'' phase of known duration
$T_{ar}$ where all the lights are red before some light passes from
red to green. These constants are very important for respecting safety
regulations requiring that when the ego vehicle enters the critical
section the lights should be either green or yellow. In addition, a
vehicle entering the intersection must exit before a light turns green
on a transverse road.

  \begin{figure}[htbp]
    \centering
    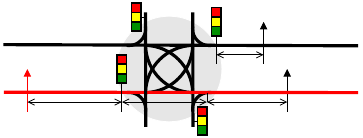
    \caption{\label{fig:crossing-traffic-light}{Crossing
  configuration with traffic lights}.}
  \end{figure}

In the caution phase, the ego vehicle must maintain its speed so
that it can stop before the traffic light, i.e.  $\Brake(\speed{e})
\le d_e$.  The clearance condition for the progress phase is twofold:
\begin{itemize}
\item First, the ego vehicle should not see red light which means even if
the lights switch to yellow right after the decision to cross is
taken, the ego vehicle will reach the entrance of the intersection
before the lights turn to red. That is $\timeTo{\cfg, \po{h},\po{f}}
\le T_y$.

\item Second, the time needed to cross the critical section, that is to reach
the position $\po{h}+_e\cd{h}$ after the junction must be less that
$T_y+T_{ar}$, that is, $\timeTo{\cfg, \po{h} +_e \cd{h},\po{f}} \le
T_y+T_{ar}$.

\item Third, the junction is clear of vehicles, that is, no arriving
  vehicle $a$ is in the junction except if located at its respective entry $h_a$
  and not moving $v_a = 0$.
\end{itemize}

The crossing with traffic lights policy is therefore
formalized as follows.
\begin{align*}
  \mathsf{cross\mbox{-}}&\mathsf{traffic\mbox{-}light}(\cfg) \eqdef \\
  \initkw & \mapsto \clf \leftarrow \mathit{false}, h \leftarrow \mathsf{traffic\mbox{-}light}(F)\\
  \dokw ~~ \neg \clf & \mapsto ~\follow{\cfg,\po{h}},~ \\
  & \clf \leftarrow \big(\mathit{color}_{h} = \mathit{green} \mbox{ and} 
  \timeTo{\cfg, \po{h},\po{f}} \le T_y \mbox{ and } \\
  & ~~ \timeTo{\cfg, \po{h} +_e \cd{h},\po{f}} \le T_y + T_{ar} \mbox{ and} \\
  & ~~ \bigwedge\nolimits_{a \in A} (p_a = p_h \Rightarrow (p_a = p_{h_a} \wedge v_a = 0) \big) \\
  \mid ~~~~ \clf & \mapsto ~\follow{\cfg,\po{f}} ~~ \odkw
\end{align*}

\subsubsection{Crossing with all-way stop}
The crossing configuration $\cfg = \langle q_e, q_F, q_A \rangle$ involves the
ego vehicle approaching an intersection protected by a stop sign $o$
as illustrated in Fig.~\ref{fig:crossing-all-way-stop}.  We assume the
intersection is \emph{all-way stop}, that is, all its entries are actually
protected by stop signs.

  \begin{figure}[htbp]
    \centering
    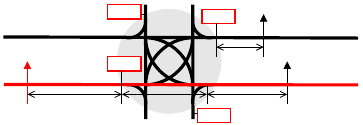
    \caption{\label{fig:crossing-all-way-stop}{Crossing configuration with all-way stop}.}
  \end{figure}

The \textsf{cross-stop} policy includes three phases.  In the first
caution phase, the ego vehicle approaches the junction such that to
stop at position of the stop sign $\po{h}$.  Then, in the second
caution phase the ego vehicle waits at the stop position until (i) the
intersection is clear of vehicles and (ii) ego is the earliest arrived
among the waiting vehicles.  Then, in the third progress phase, the ego
vehicle crosses the intersection by following its leading front vehicle $f$
from $F$.

In order to specify and implement the condition (ii) above, we assume
that $q_e$ includes a time variable $st_e$ (\textit{stop time}) that
records the exact time when the vehicle stops.  Also, we remind that the $\crossing$
relation holds for positions belonging to the same junction, as
introduced in subsection~\ref{sec:approach:maps}.

The \textsf{cross-stop} policy is therefore formalized as follows:
\begin{align*}
  \mathsf{cross\mbox{-}stop}(\cfg) \eqdef ~\initkw \mapsto & ~\clf \leftarrow \mathit{false}, h \leftarrow \mathsf{stop\mbox{-}sign}(F) \\
  \dokw~~ \po{e} <_e \po{h} \mapsto & ~\follow{\cfg,\po{h}} \\
  \mid~~ \po{e} = \po{h} \wedge \neg \clf \mapsto & ~\follow{\cfg,\po{h}}, \\
  \clf \leftarrow \bigwedge\nolimits_{a \in A} & \big(\po{e} \crossing \po{a} \implies (\po{a} = \po{h_a} \wedge \\
  & \speed{a} = 0 \wedge st_e < st_a)\big) \\
  \mid \hspace{2cm} \clf \mapsto & ~\follow{\cfg,\po{f}} ~~\odkw
\end{align*}

Note that we tacitly exclude the possibility to observe the same stop
time for two vehicles and hence to mutually block the progress over
the junction.  This can be eventually achieved by using an additional
arbitration protocol e.g., a total priority order between junction
entries.

\section{Correctness of the Method}\label{sec:correctness}

We will use a compositional approach for guaranteeing safety,
regardless the number of vehicles and the map characteristics.  For
doing that, we first define a concept of \emph{free space} for
vehicles depending on their current driving configuration and control
policy.  Then, we show that the vehicles can adapt their speed
depending on their free space, that is, are always driving safely
while remaining within their free space.  Moreover, we show that the
proposed configuration management and control policies guarantee the free
spaces stay disjoint throughout the execution, and thus ensure the
absence of collision between vehicles at any time.

The proof of correctness is, however, subject to few additional mild
assumptions, listed here for the sake of completeness:
\begin{enumerate}
\item \label{as:simple-maps} The merging points, crossing points and
  the junctions of the map are neatly separated by road segments.
  Thus, the frontal visibility of every vehicle contains at most one
  critical position i.e., a merging, a crossing or a junction
  entry-point.  Also, no such critical positions exist on the arriving
  roads, with the exception of the entry points belonging to the same
  junction, if any.  In particular, this assumption implies that the
  transition from one policy to another is always to/from the
  \textsf{road} policy.  Thus, the behavior of the Control Policy
  Manager can be modeled as a mode automaton where each mode
  corresponds to the application of one particular type of policy, as
  illustrated in Fig.~\ref{fig:mode-policies}.
\item \label{as:visibility-non-retracting} The frontal and lateral
  visibility do not retract. That is, $p_e +_e fd(q_e) \le_e p_e'
  +_e fd(q_e')$ and $ld_r(q_e) \le_e ld_r(q_e')$ for every vehicle
  $e$, at any consecutive positions and states $p_e,p_e',q_e,q_e'$
  respectively, for every arriving road $r$ in its configuration.  Moreover, a
  vehicle located at some entry point of a junction has full
  visibility over all the positions of that junction.
\item \label{as:speed-limits-consistent} Speed limit signals are neatly
  separated and moreover, the limits never increase when approaching junctions
  and/or any critical road segments.  This assumption is
  mandatory for the conservative evaluation of the clearance
  conditions.
\item \label{as:reactive} The time period $\Delta t$ is small enough to
  guarantee $v \cdot \Delta t \le \Brake(v)$, for all $v \ge V_0 > 0$
  where $V_0$ is the minimal speed limit used on roads. For
  example, if $V_0=5km/h$ and $\Brake(v) \sim v^2/(-2b_{max})$ for
  $b_{max} = -3.4m/s^2$ this assumption reduces to $\Delta t \le
  200ms$.  The assumption is mandatory for guaranteeing the vehicles
  remain in their free spaces at every cycle, as defined next.
\end{enumerate}

  \begin{figure}[htbp]
    \centering
    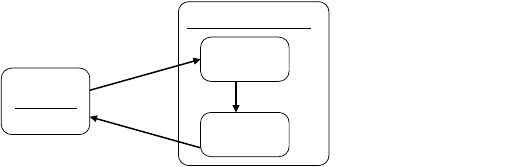
    \caption{\label{fig:mode-policies}{Refined behavior of the Control Policy Manager}.}
  \end{figure}

We remind that the behavior of a SADS is characterized by all the
possible execution sequences $\langle q[t] \rangle_{t = n\Delta t}$
consisting of successive states reached every $\Delta t$.
We assume that vehicles operate according to their specific configuration
policies, as introduced previously.  That is, at every step $\Delta
t$, every vehicle $e$ applies the corresponding caution or progress
policy step corresponding to its configuration $\cfg$.

\subsection{Safe configurations and safe states}

Let $q \eqdef q[t]$ be a SADS state and $e$ be a vehicle.  Let
respectively $\cfg \eqdef \cfg_e(q) \eqdef \langle q_e, q_F, q_A \rangle$
be the current configuration (of some type $T$) and $\mathsf{pol}_T(\cfg)$ be
the control policy applied by $e$ in $\cfg$.  Consider
$q_e \eqdef \langle \po{e}, s_e, \speed{e}, V_e, ... \rangle$.  Then,
we define
\begin{itemize}
\item the \emph{front lead obstacle} $\fo{e}$ as the obstacle from
$F$ whose position is followed by $e$ in the current execution step of
$\mathsf{pol}_T(\cfg)$, that is, either
\begin{itemize}
\item the signal $h$ protecting the merging or crossing (e.g., yield sign,
  traffic light, stop sign, etc) during the caution phase, or
\item the last (closest) front vehicle $f$ during the progress phase of
  $\mathsf{pol}_T(\cfg)$
\end{itemize}
\item the \emph{limit position} 
\( \limit{e} \eqdef \min\nolimits_e
(\po{\fo{e}}, \po{e} +_e \Brake(V_e), (\po{h} +_e \Brake(V_h))_{h \in
F_{sl}})\) where $F_{sl} \eqdef \{h \in F ~|~ \att{h}{type}=\textsf{speed-limit},~\po{h} <_e \po{\fo{e}} \}$
\item the \emph{free space} of $e$ as the route 
$fs_e \eqdef \ride{}{\po{e}}{s_e[0,\ell]}{\limit{e}}$ where
$\ell \eqdef \limit{e} - \po{e}$.
\end{itemize}
In other words, the limit of free space is determined by the strongest
constraint resulting either from the position of the front obstacle,
or from the distances required to comply with speed limits.  Let us
observe that $\po{e} \le_e \limit{e} \le_e \po{\fo{e}} \le_e \po{e}
+_e fd(q_e)$, that is, the free space never exceeds the position of
the front lead obstacle, and consequently of the frontal visibility
limit of $e$ (because the obstacle $\fo{e}$ is one of the front
obstacles from $F$, visible in the configuration $\cfg$, see
Fig.~\ref{fig:free-space}).

  \begin{figure}[htbp]
    \centering
    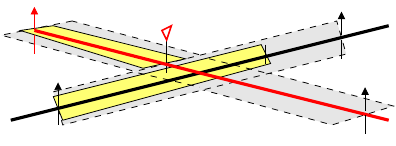
    \caption{\label{fig:free-space}{Examples of configurations (gray) and free spaces (yellow)}.}
  \end{figure}

We call a driving configuration $\cfg$ of type $T$ \emph{safe} if $I_1(\cfg) \wedge
I_2(\cfg)$ hold where
\begin{align*}
I_1(\cfg) \eqdef & (\po{e} +_e \Brake(\speed{e}) \le_e \limit{e}) \\
I_2(\cfg) \eqdef & \left\{ \begin{array}{ll}
  \mathsf{true} & \mbox{if } T=\mathsf{road} \\
  (\fo{e} = h \wedge \po{e} \le_e \po{h}) \bigvee \\
  ~~ (\po{h} +_e \cd{h} <_e \po{\fo{e}} \wedge \mathcal{C}_T) & 
  \mbox{otherwise} \\
\end{array} \right.
\end{align*}
and $\mathcal{C}_T$ is provided in Table~\ref{tab:vista-conditions}, for different types $T$.

$I_1(\cfg)$ ensures the speed of the $e$ vehicle is adapted to its free
space in its configuration $\cfg$.  $I_2(\cfg)$ ensures the lead obstacle
$\fo{e}$ is consistently defined with respect to the execution phase
of the policy $\mathsf{pol}_T(\cfg)$ for merging or crossing configuration $\cfg$.
The two terms of the disjunction correspond to the caution and the
progress phases respectively.  For the caution phase, the safety
condition ensures that the free space is limited by the protecting
signal $h$ and hence not intersecting the critical section of the
junction or merging. In the progress phase, however, the free space
may expand over the critical section.  In this case, the condition
$\mathcal{C}_T$ provides the additional guarantees for having
non-intersecting free spaces.

\begin{table}[htbp]
\caption{\label{tab:vista-conditions} Specific safety conditions for merging and crossing configurations}
\begin{tabular}{|p{.26\columnwidth}| p{.64\columnwidth} | } \hline
~~$T$ & ~~$\mathcal{C}_T$ \\ \hline
\textsf{merge/cross-yield} &
$V_a \cdot \timeTo{\cfg,\po{h} +_e \cd{h},\po{\fo{e}}} + \Brake(V_a) \le d_a$ \\[5pt]
\textsf{cross-traffic-light} & $\po{e} \le_e \po{h} \implies (\timeTo{\cfg,\po{h},\po{\fo{e}}} < \timeToRed{h}) \wedge$ \\
& $\bigwedge_{a \in A} \timeTo{\cfg, \po{h}+_e \cd{h}, \po{\fo{e}}} < \timeToGreen{h_a} \wedge$ \\
& $\bigwedge_{a \in A} \po{a} \crossing \po{h} \implies (\po{a} = \po{h_a} \wedge \speed{a} = 0)$ \\[5pt]
\textsf{cross-stop} & $\bigwedge_{a \in A} \po{a} \crossing \po{h} \implies 
(\po{a} = \po{h_a} \wedge \speed{a} = 0 \wedge$ \\
& $\att{e}{st} < \att{a}{st})$ \\ \hline
\end{tabular}
\end{table}

We call a state $q$ \emph{safe} if all the driving configurations are safe, for
all vehicles $e$.
We call a state $q$ \emph{compliant to speed limits} if
$\speed{e} \le V_e$ for all vehicles $e$.

\begin{lemma}\label{lemma:safe-is-speed-compliant-and-disjoint}
A safe state is compliant to speed limits and ensures disjoint free
spaces for all vehicles.
\end{lemma}
\begin{proof}
As the driving configuration of every vehicle $e$ is safe it holds $\po{e}
+ \Brake(\speed{e}) \le_e \limit{e}$.  By definition of the limit
position it holds $ \limit{e} \le_e \po{e} +_e \Brake(V_e)$.  Hence,
$\Brake(\speed{e}) \le \Brake(V_e)$ and so $\speed{e} \le V_e$.
Therefore, the state is compliant to speed limits.

By contradiction, suppose two vehicles $c_1$, $c_2$ have
intersecting free spaces.  First, consider that $\po{c_2}$ belongs
to the free space of $c_1$, that is, $\po{c_1} <_{c_1} \po{c_2}
<_{c_1} \limit{c_1}$.  This implies $c_2$ belongs to the front
obstacles of $c_1$ which is moreover located before $\fo{c_1}$ because
$\limit{c_1} \le_{c_1} \po{\fo{c_1}}$.  But this is impossible by the
definition of $\fo{c_1}$ which is either the nearest protecting signal
$h$ or the nearest front vehicle to $c_1$.

Second, consider there exists a common position $p$ belonging to the
two free spaces, that is, $\po{c_1} <_{c_1} p <_{c_1} \limit{c_1}$ and
$\po{c_2} <_{c_2} p <_{c_2} \limit{c_2}$.  This can happen only if $p$
is a merging or crossing position of the routes of $c_1$ and $c_2$.
Let $T_1$, $T_2$ be the types of driving configurations of $c_1$, $c_2$ respectively.
We consider the following cases, the others being symmetric:
\begin{itemize}
\item $T_1=$ \textsf{merge-yield}, $T_2=$ \textsf{road}:
that is, $c_1$ must give way to $c_2$ on the merging at position $p$.
As the free space of $c_1$ extends beyond $p$ that means $c_1$ is in
the progress phase, thus, following a front vehicle $\fo{c_1}$ located
beyond the critical section, that is, after $\po{h} +_{c_1} \cd{h}$.
Also, $c_2$ is perceived as an arriving vehicle in the configuration of $c_1$.
Since the merging configuration of $c_1$ is safe it satisfies the associated
invariant, hence we have $V_{c_2} \cdot \timeTo{\cfg_{c_1},\po{h}
+_{c_1} \cd{h},\po{\fo{c_1}}} + \Brake(V_{c_2}) \le d_{c_2})$.  This
implies  $\Brake(V_{c_2}) \le d_{c_2}$ and consequently
$\limit{c_2} <_{c_2} p$, thus contradicting the assumption that the
two free spaces intersect.
\item $T_1=$ \textsf{cross-yield}, $T_2=$ \textsf{road}: identical
to the previous case.
\item $T_1=T_2=$ \textsf{cross-traffic-light}: the free spaces could
intersect only if the two vehicles $c_1$ and $c_2$ are in the progress
phase.  If at least one of them is in the junction (i.e., $c_1$) then
the safety condition for the other (i.e. $c_2$) is violated.  That is,
the 3rd constraint fails for the configuration safety of $c_2$ whenever
$\po{c_1} \crossing \po{h}$.  If none of $c_1$ or $c_2$ has yet
entered the junction, their safety conditions would require for both
of them to have the traveling time to their respective traffic lights
lower than the time to pass to red, that is, the 1st constraint.  But
this would mean both of the two traffic lights have a non-red color,
which is forbidden.
\item $T_1=T_2=$ \textsf{cross-stop}: the free spaces could intersect
only if the two vehicles $c_1$ and $c_2$ are in the progress phase.
But then, if at least one of them is already in the junction, the
safety condition of the other one is violated, as in the previous
case.  Also, if both are at their entry points, the safety conditions
of their configurations would require for both of them to have the lowest stop
time, which is impossible.
\end{itemize}
\end{proof}

\subsection{Preservation of safe configurations}

Let respectively $q \eqdef q[t]$, $q' \eqdef q[t + \Delta t]$ be two
consecutive SADS states.  Let $e$ be a vehicle with states
$q_e \eqdef \langle \po{e}, s_e, \speed{e}, V_e, ... \rangle$,
$q_e' \eqdef \langle \po{e}', s_e', \speed{e}', V_e', ... \rangle$ and
configurations $\cfg \eqdef \cfg_e(q) \eqdef \langle q_e, q_F, q_A\rangle$,
$\cfg' \eqdef \cfg_e(q') \eqdef \langle q_e', q_{F'}', q_{A'}' \rangle$ at
$q$ and $q'$, respectively.  Let $T$, $T'$ be the types of the two
configurations above and $\mathsf{pol}_T(\cfg)$, $\mathsf{pol}_{T'}(\cfg')$ the
corresponding control policies.

Let respectively $\fo{e}$, $\fo{e}'$ be the front lead obstacles of
$e$ according to its configurations $\cfg$, $\cfg'$.  We call the step from $q$ to
$q'$ \emph{non-intrusive for} $e$ if either
$\po{\fo{e}} \le_e \po{\fo{e}'}$ or $\po{e}'
+_e \Brake(V_e') \le_e \po{\fo{e}'}$.  That is, either the position of
the front obstacle is progressing on the route of $e$, or it can
change arbitrarily as long as its distance to $\po{e}'$ is greater
than $\Brake(V_e')$.

We are now ready for proving two key preservation properties.  The
next two lemmas give the conditions for preservation of configuration safety,
assuming respectively, the types of the configurations do not or do change
between $q$ and $q'$.

\begin{lemma}\label{lemma:preserve-safe-vista-type-unchanged}
If the configuration $\cfg$ is safe, the states $q$ and $q'$ are compliant to
speed limits, the step from $q$ to $q'$ is non-intrusive for $e$ and
$\cfg'$ has the same type as $\cfg$ then, the configuration $\cfg'$ is safe.
\end{lemma}
\begin{proof}
\underline{$I_1(\cfg')$}. 
The condition $\po{e} +_e \Brake(\speed{e}) \le \limit{e}$ is
actually equivalent to conditions
\ref{it1:prop:speed-controllability}--\ref{it3:prop:speed-controllability} of
Proposition \ref{prop:speed-controllability} when considering the
(lead) position $p \eqdef \po{\fo{e}}$.  Then,
\begin{itemize}
\item conditions
  \ref{it5:prop:speed-controllability}--\ref{it6:prop:speed-controllability}
  of Proposition \ref{prop:speed-controllability} guarantee $\speed{e}' \le
  V_e'$ depending on the current applicable speed limit
\item condition \ref{it7:prop:speed-controllability} of
  Proposition \ref{prop:speed-controllability} guarantees $\po{e}' +_e
  \Brake(\speed{e}') \le \po{h} + \Brake(V_h)$ for the speed limits $h
  \in F_{sl}$ such that $\po{e}' \le_e \po{h} <_e \po{\fo{e}}$
\item condition \ref{it8:prop:speed-controllability} of
  Proposition \ref{prop:speed-controllability} guarantees $\po{e}' +_e
  \Brake(\speed{e}') \le_e \po{\fo{e}}$
\item the step from $q$ to $q'$ being non-intrusive for $e$ guarantees
  $\po{\fo{e}} \le_e \po{\fo{e}'}$ or $\po{e}' + \Brake(V_e') \le_e
  \po{\fo{e}'}$.
\end{itemize}
Using all the above we can infer $I_1(\cfg')$ that is
\[\po{e}' +_e \Brake(\speed{e}') \le_e \limit{e}' \eqdef
\min\nolimits_e(\po{\fo{e}'}, \po{e}' +_e \Brake(V_e'), (\po{h} +_e
\Brake(V_h))_{h \in F'_{sl}}\]
where $F_{sl}' = \{ h \in F' ~|~ h \mbox{ speed limit}, \po{h} <
\po{\fo{e}'}\}$.  In particular, observe that for any speed limit $h
\in F'_{sl}$ that is not taken into account in $F_{sl}$ we have
$\po{e}' +_e \Brake(\speed{e}') \le_e \po{\fo{e}} <_e \po{h} <_e
\po{h} + \Brake(V_h)$.

\underline{$I_2(\cfg')$} First, assume $\fo{e} = h \wedge \po{e} \le_e
\po{h}$, that is, the policy $\mathsf{pol}_T(\cfg)$ is in the caution
phase for configuration $\cfg$.  Two situations can happen during the step:
\begin{itemize}
\item \emph{no clearance:} then, at the next configuration $\cfg'$ we still have
  $\fo{e}' = h$ and we know from the previous point that $\po{e}' +_e
  \Brake(\speed{e}') \le_e \po{\fo{e}}$.  Then, we obtain immediately
  $\po{e}' \le_e \po{h}$, that is, $I_2(\cfg')$ holds for the caution
  phase of $\mathsf{pol}_T(\cfg')$.
\item \emph{clearance}: then, at the next configuration $\cfg'$ we will have
  $\fo{e}' = f$ for some front vehicle $f$ located beyond the critical
  section, that is, $\po{h}+_e \cd{h} \le_e \po{f}$.  Moreover, as the
  clearance condition holds at $\cfg$ for $\mathsf{pol}_T(\cfg)$ we can
  check that the specific condition $\mathcal{C}_T$ holds at $\cfg'$,
  for every type of merging or crossing configuration $T$.  That is, we obtain
  that $I_2(\cfg')$ holds for the progress phase of
  $\mathsf{pol}_T(\cfg')$.
\end{itemize}

Second, assume $\po{h} +_e \cd{h} \le_e \po{\fo{e}} \wedge
\mathcal{C}_T$, that is, the policy $\mathsf{pol}_T(\cfg)$ is in the
progress phase at configuration $\cfg$.  Then, we obtain that $\po{h} +_e \cd{h}
\le_e \po{\fo{e}'}$ if $\po{\fo{e}} \le_e \po{\fo{e}'}$.  This latter
condition holds if the step is non-intrusive and moreover, assuming
there are no entry points located immediately after the end of the
critical section, that is, Assumption~(\ref{as:simple-maps}).  Also, we
can check that the specific condition $\mathcal{C}_T$ holds at $\cfg'$
as well.  In particular, we use that the states $q$, $q'$ are speed
compliant and moreover the speed limits are non-increasing when
approaching critical road segments, that is,
Assumption~(\ref{as:speed-limits-consistent}).  In addition, we use the
following monotonicity property of $\timeTo{}$:
\[\timeTo{\cfg', p, \po{\fo{e}'}} \le \timeTo{\cfg, p, \po{\fo{e}}} - \Delta t\]
for any position $p$ between $\po{e}$ and $\po{\fo{e}}$ provided
$\po{\fo{e}} \le_e \po{\fo{e}'}$.  That is, the time to reach a
position can only improve due to the moving forward of the front
obstacle.
\end{proof}

\begin{lemma} \label{lemma:preserve-safe-vista-type-changed}
If the configuration $\cfg$ is safe, the states $q$ and $q'$ are compliant to
speed limits, the step from $q$ to $q'$ is non-intrusive for $e$ and
$\cfg'$ has a different type than $\cfg$ then, the configuration $\cfg'$ is safe.
\end{lemma}
\begin{proof}
  \underline{$I_1(\cfg')$}: as in the previous
  lemma. \underline{$I_2(\cfg')$}: First, consider $T=\textsf{road}$ and
  $T'\not=\textsf{road}$.  The change from the \textsf{road} policy to
  any other merging or crossing policy is performed because the first
  visible obstacle changes from a front vehicle $f$ (in $\cfg)$ to a
  signal $h$ (in $\cfg'$).  This implies $\po{f} <_e \po{h}$, otherwise
  $\cfg$ would not have been a \textsf{road} policy.  We know for the
  \textsf{road} policy that $\po{e}' \le_e \po{f}$, that is, the $e$
  vehicle must be able to stop behind the front vehicle $f$.
  Consequently, $\fo{e}' = h \wedge \po{e}' \le_e \po{h}$ holds in
  $\cfg'$, that is, $I_2(\cfg')$ holds on the first case.

  Second, consider $T\not=\textsf{road}$ and $T'\not=T$.  That is, the
  $e$ vehicle is in a merging or crossing configuration $\cfg$ and, while being
  in the progress phase, it goes beyond the signal $h$ and receives
  another configuration $\cfg'$ at the next step.  By
  Assumption~(\ref{as:simple-maps}) we know that all critical regions of
  the maps (that is, mergings, crossings, junctions, etc) are
  separated by road segments.  Henceforth, $T'=\textsf{road}$ unless
  the $e$ vehicle is perceiving another signal $h$ immediately after
  treating the configuration $\cfg$, which is impossible. Therefore, $I_2(\cfg')$ is
  trivially satisfied as $\cfg'$ is a \textsf{road} configuration.
\end{proof}

The next lemma gives a property about the evolution of the limit
position and the derived free space in a step.

\begin{lemma}\label{lemma:safe-vista-non-retracting-free-space}
If the configuration $\cfg$ is safe, the step from $q$ to $q'$ is non-intrusive
for $e$, then $\limit{e} \le_e \limit{e}'$, that is, the free space of
$e$ is \emph{non-retracting} between $q$ and $q'$.
\end{lemma}
\begin{proof}
In general, recall that $\min X \le \min Y$ iff $\forall y \in
Y~\exists x \in X.~x \le y$.  Therefore, let us consider the terms
used in the min-definition of $\limit{e}'$, respectively
\begin{itemize}
\item $\po{\fo{e}'}$: as the step from $q$ to $q'$ is non-intrusive
we know that $\po{\fo{e}} \le_e \po{\fo{e}'}$ or $\po{e}'
+_e \Brake(V_e') \le_e \po{\fo{e'}}$.  In the first case we are done (as
$\po{\fo{e}}$ occurs in the min definition of $\limit{e}$).  In the second
case we are also done as $\po{\fo{e}'}$ is not the minimal value for the
definition of $\limit{e}'$.
\item $\po{e}' +_e \Brake(V_e')$: we have either $V_e = V_e'$, that is,
the speed limit does not change, or $V_e \not= V_e'$.  In the first
case, $\po{e} +_e \Brake(V_e)$ occurs in the min definition of
$\limit{e}$ and moreover $\po{e} \le_e \po{e}'$.  In the second case,
$V_e' = V_h$ for some speed limit $h \in F_{sl}$ visible for the ego
vehicle $e$ in state $q$. But then, $p_h +_e \Brake(V_h)$ occurs in
the min definition of $\limit{e}$ and moreover
$\po{e} \le_e \po{h} \le_e \po{e}'$.  Hence, $\po{h}
+_e \Brake(V_h) \le_e \po{e}' +_e \Brake(V_e')$.
\item $\po{h} +_e \Brake(V_h)$ for $h \in F'_{sl}$:
we have either $h \in F_{sl}$ or $h \in F'_{sl} \setminus F_{sl}$.  In
the first case, we are done, as the same term occurs on both sides.  In
the second case, we have $\po{\fo{e}} \le_e \po{h} <_e \po{h}
+ \Brake(V_h)$.
\end{itemize}
\end{proof}

\subsection{Preservation of state safety}

Let respectively $q \eqdef q[t]$, $q' \eqdef q[t + \Delta t]$ be two
consecutive SADS states.  We say that the step from $q$ to $q'$
\emph{has no gaps} if for any vehicle $e$, $\po{e}' \in fs_e$, that
is, the $e$ vehicle moves within its free space.

\begin{lemma} \label{lemma:no-gaps}
  If $q$ is safe then, the step from $q$ to $q'$ has no gaps and is
  collision-free.
\end{lemma}
\begin{proof}
  Let fix some vehicle $e$.  Remember that the vehicle $e$ is following the
  front leading obstacle $\fo{e}$ in the current execution step.
  Therefore, it must be able to stop at $\fo{e}$.  We distinguish two
  cases.  First, if the limit position $\limit{e}$ is equal to
  $\po{\fo{e}}$ the vehicle must therefore be able to stop at
  $\limit{e}$ hence, it moves within its free space.  Otherwise, the
  limit position $\limit{e}$ is strictly smaller than $\po{\fo{e}}$ and
  we further distinguish two sub-cases:
  \begin{itemize}
  \item $\limit{e} = \po{e} +_e \Brake(V_e)$: then, the maximal speed
    of $e$ during the step could be $V_e \ge V_0$, and using
    Assumption~(\ref{as:reactive}) we obtain $V_e \cdot \Delta t \le
    \Brake(V_e)$ hence, the $e$ vehicle moves in its free space
  \item $\limit{e} = \po{h} +_e \Brake(V_h)$ for some speed limit $h$
    located ahead on the route of $e$: then, the maximal speed of $e$
    during this step could be $V' \ge V_h \ge V_0$ such that
    $\Brake(V') = (\po{h} +_e \Brake(V_h)) - \po{e}$.  Again, using
    Assumption~(\ref{as:reactive}) we obtain $V' \cdot \Delta t \le
    \Brake(V')$ which implies that the vehicle $e$ moves within its
    free space.
  \end{itemize}
  In conclusion, the step has no gaps.  As $q$ is safe, by
  Lemma \ref{lemma:safe-is-speed-compliant-and-disjoint} the free spaces
  are disjoint.  Hence, the step is collision-free.
\end{proof} 

\begin{lemma} \label{lemma:safe-is-non-intrusive}
If $q$ is safe then, the step from $q$ to $q'$ is non-intrusive for
all vehicles.
\end{lemma}
\begin{proof}
  By contradiction, assume the step from $q$ to $q'$ is intrusive for
  some vehicle $e$, that is, we $\po{\fo{e}'} <_e \po{\fo{e}}$ and
  $\po{\fo{e}'} <_e \po{e}' +_e \Brake(V_e')$.  This means that the front
  leading obstacle has shifted backwards and is moreover located below
  the speed limit range $\po{e}' +_e \Brake(V_e')$.

  We shall distinguish two cases, depending on the nature of
  $\fo{e}'$.  First, consider that $\fo{e}'$ is a signal $h$.  But then,
  $h$ was already in the visibility range at state $q$ and, as $\po{h}
  < \po{\fo{e}}$ it has been already considered (i.e.,
  cleared) in the caution phase of the policy.
  
  Second, consider that $\fo{e}'$ is a vehicle.  If $\fo{e}'$ is the
  fictitious vehicle located at the front visibility limit then, the
  visibility limit has retracted from the previous state, where the
  position $\po{\fo{e}}$ was visible.  This contradicts
  Assumption~(\ref{as:visibility-non-retracting}).  Otherwise, as
  vehicles are never moving backwards on their routes, $\fo{e}'$ must
  be some new vehicle that enters or crosses the road of the $e$
  vehicle in the new configuration $\cfg'$ of $e$ at $q'$.  Then, at state $q$,
  the vehicle $e$ is approaching a merging, crossing or junction and
  as $\po{fo_e'} < \po{\fo{e}}$, it must be in a progress
  phase by Assumption~(\ref{as:simple-maps}).  We distinguish the
  following situations depending on the configuration type $T$ of $e$ at at
  state $q$:
  \begin{itemize}
  \item $T=\textsf{road}$: then, the vehicle $e$ is driving on a
    high-priority road, which is crossed or joined by some
    low-priority road.  If some other vehicle $\fo{e}'$ appears in 
    front of the $e$ vehicle such that $\po{\fo{e}'} <_e \po{e}' +_e
    \Brake(V_e')$ that means $\fo{e'}$ would be also in the progress
    phase of its configuration of type \textsf{merge-yield} or
    \textsf{cross-yield}.  But then, this contradicts the safety of
    this configuration of $\fo{e}'$.
  \item $T=\textsf{merge-yield}$ or $T=\textsf{cross-yield}$: this
    situation is the dual of the above and we have a similar
    contradiction for the safety of the configuration of $e$, which must not
    be in the progress phase while another (high-priority arriving)
    vehicle (that is, $\fo{e}'$) could enter on its road.
  \item $T=\textsf{cross-traffic-light}$ or $T=\textsf{cross-stop}$:
    again, these situations contradict the safety of the configuration of $e$,
    which means that no other vehicle may cross the junction at the same
    time.
  \end{itemize}  
\end{proof}

\begin{lemma} \label{lemma:safe-inductive}
If $q$ is safe then $q'$ is speed compliant and safe.
\end{lemma}
\begin{proof}
  First, by Lemma \ref{lemma:safe-is-speed-compliant-and-disjoint} the
  state $q$ is speed compliant.
  
  Second, we prove that state $q'$ is also speed compliant.  Consider an arbitrary
  $e$ vehicle and its safe configuration $\cfg$.  As in
  Lemma \ref{lemma:preserve-safe-vista-type-unchanged}, the condition
  $\po{e} + \Brake(\speed{e}) \le \limit{e}$ guarantees the conditions
  \ref{it1:prop:speed-controllability}--\ref{it3:prop:speed-controllability}
  of Proposition \ref{prop:speed-controllability} when the $e$ vehicle is
  following the front leading obstacle $\fo{e}$.  Then, the conditions
  \ref{it5:prop:speed-controllability}--\ref{it6:prop:speed-controllability}
  guarantee $v_e' \le V_e'$, depending on the applicable speed limit
  $V_e'$ at the next state $q'$.  

  Third, by Lemma \ref{lemma:safe-is-non-intrusive} the step from $q$ to
  $q'$ is non-intrusive for all vehicles.  We just proven that $q$ and
  $q'$ are speed compliant.  Then, as all configurations $\cfg$ are safe at $q$
  by Lemma \ref{lemma:preserve-safe-vista-type-unchanged} and
  Lemma \ref{lemma:preserve-safe-vista-type-changed} all configurations $\cfg'$
  are safe at $q'$.  Hence, $q'$ is safe.
\end{proof}
 
\begin{theorem}
If the SADS is initially safe, it will remain so throughout its
execution. That is, the execution will preserve safety by avoiding
collisions and complying with applicable traffic regulations.
\end{theorem}
\begin{proof}
If the SADS is initially safe, by Lemma \ref{lemma:safe-inductive} all
the states along the execution are safe.  Moreover, by
Lemma \ref{lemma:safe-is-speed-compliant-and-disjoint} these states are
compliant to speed limits and collision free.   Furthermore, by
Lemma \ref{lemma:no-gaps} the steps between them are
collision-free.
\end{proof}

\section{Discussion}\label{sec:discussion}
This paper makes an original and fundamental contribution in a
direction rarely explored for ADS, where the state of the art in
autopilot design focuses on two diametrically opposed approaches, each
with inherent limitations in guaranteeing their safety. On the one
hand, AI-based end-to-end solutions cannot provide the necessary
safety guarantees. On the other hand, in addition to the inherent
limits of complexity, traditional control-theoretic approaches
over-simplify the problem to make it accessible to mathematical
analysis, and are far from being able to realistically take into
account the important details associated with different driving
operations. The proposed method makes it possible to construct a safe
rule-based autopilot from configurations representing the perceived
state of the vehicle environment. The method offers an integrated
solution for the three aspects of prediction, planning and control,
and allows direct implementation involving simple calculations.

The idea that driving a vehicle boils down to a composition of skills,
each dealing with specific situations, is widely held. Advocated by
\cite{Albus22}, it is also adopted in validation techniques that focus
on particular classes of pre-crash scenarios involving high-risk
operations \cite{NajmSY07}. The paper demonstrates the advantages of
compositional reasoning, breaking down the general problem into
sub-problems that can be tackled by successive realistic
simplifications.

A first class of simplifications is based on the application of rules
whose scope is limited by context and knowledge, to which is added the
principle of rights-based responsibility. This leads to the definition
of the notion of driving configuration for each vehicle, and the
corresponding free space in which it can circulate in complete safety.
The second class of simplifications comes from the fact that analysis
of autopilot inputs enables classification into a very limited number
of types of configurations, each characterized by a corresponding
control policy. Hence, the compositionality result according to
which safe driving for each type of configuration implies safe driving for any
route.  Finally, a third class of simplifications derives from the
assumption that vehicles drive responsibly, strictly observing the
traffic rules and staying within their allocated free space in all
cases. In this way it is possible to envisage minimal configurations
for each vehicle, including a single vehicle in front of it and
vehicles whose routes may cross its own. This greatly simplifies the
mathematical analysis, which deduces the invariant constraints to be
respected by control policies on the basis of knowledge of two
functions characterizing a vehicle's controllability. These functions
and their properties define a kind of contract between the autopilot
and the underlying electrical control system for braking and steering.
They are an essential element of the proposed solution, as they
provide the predictability on which safe control policies can be
built.

The approach presented rests on solid theoretical foundations and
highlights a third way of constructing ADS that deserves to be
explored further. On the one hand, it should be implemented and
validated in a simulation environment. On the other hand, it should be
extended to the study of safety policies for other operations such as
U-turns and parking maneuvers.  In addition, their application needs
to be refined by considering more detailed maps with two-dimensional
road representation and the integration of autopilot with trajectory
control modules.

\bibliographystyle{splncs04}
\bibliography{biblio-compact}

\end{document}